\documentclass[a4paper,USenglish,cleveref, autoref, thm-restate]{lipics-v2021}
\usepackage{algpseudocode}
\usepackage[dvipsnames]{xcolor}
\usepackage{caption}

\algtext*{EndWhile}% Remove "end while" text
\algtext*{EndIf}% Remove "end if" text
\algtext*{EndFor}% Remove "end for" text
\algtext*{EndProcedure}
\algtext*{EndFunction}

\newtheorem{lossyredrule}{Lossy Reduction Rule}
\newtheorem{lossyredprotocol}[lossyredrule]{Lossy Reduction Protocol}

\usepackage{bm}

\newcommand{\eps}{\varepsilon}

\newcommand{\Oh}{\mathcal{O}}

\newcommand{\opt}{\ensuremath{\mathsf{opt}}}
\newcommand{\sub}{\subseteq}
\newcommand{\sm}{\setminus}

\newcommand{\nn}{\ensuremath{\mathbb{N}}}

\newcommand{\fcal}{\ensuremath{\mathcal{F}}}
\newcommand{\hcal}{\ensuremath{\mathcal{H}}}

\newcommand{\lcauv}{{\mathsf{LCA}}}
\newcommand{\lca}{\overline{\mathsf{LCA}}}

\newcommand{\fdel}{\textnormal{\textsc{$\mathcal{F}$-Deletion}}}
\newcommand{\twetadelsmall}{\textnormal{\textsc{Tw-$\eta$-Del}}}
\newcommand{\twetadel}{\textnormal{\textsc{Treewidth-$\eta$-Deletion}}}
\newcommand{\fdelsmall}{\textnormal{\textsc{$\mathcal{F}$-Del}}}

\newcommand{\val}{\textnormal{\texttt{val}}}
\newcommand{\planar}{\textnormal{\texttt{planar}}}
\newcommand{\clique}{\textnormal{\texttt{simp}}}
\newcommand{\nonclique}{\textnormal{\texttt{non-simp}}}

\newcommand{\flow}{\textnormal{\texttt{flow}}}

\newcommand{\tw}{\textnormal{\texttt{tw}}}

\title{Protrusion Decompositions Revisited: \\ Uniform Lossy Kernels for Reducing Treewidth and Linear Kernels for Hitting Disconnected Minors}
\titlerunning{Protrusion Decompositions Revisited}

\author{Roohani Sharma} {Discrete Mathematics Group, Institute for Basic Science (IBS), Daejeon, South Korea}{roohani@ibs.re.kr}{https://orcid.org/0000-0003-2212-1359}{Supported by the Young Scientist Fellowship of the Institute for Basic Science (IBS-R029-Y8).}

\author{Michał Włodarczyk}
{University of Warsaw, Institute of Informatics, Poland}
{michal.wloda@gmail.com}
{https://orcid.org/0000-0003-0968-8414}
{Supported by the Polish National Science Centre SONATA-19 grant 2023/51/D/ST6/00155.}

\authorrunning{R. Sharma, M. Włodarczyk} % mandatory. First: Use abbreviated first/middle names. Second (only in severe cases): Use first author plus 'et al.'

\Copyright{Roohani Sharma and Michał Włodarczyk} % mandatory, please use full first names. LIPIcs license is "CC-BY";  http://creativecommons.org/licenses/by/3.0/

 % mandatory; please add comma-separated list of keywords

\category{} %optional, e.g. invited paper

\relatedversion{} %optional, e.g. full version hosted on arXiv, HAL, or other respository/website
\acknowledgements{We thank the anonymous reviewer for pointing out a way to derive a finer bound for Lemma~\ref{lem:non-clique}.}%optional

\nolinenumbers %uncomment to disable line numbering
\hideLIPIcs

\EventEditors{Meena Mahajan, Florin Manea, Annabelle McIver, and Nguy\~{\^{e}}n Kim Th\'{\u{a}}ng}
\EventNoEds{4}
\EventLongTitle{43rd International Symposium on Theoretical Aspects of Computer Science (STACS 2026)}
\EventShortTitle{STACS 2026}
\EventAcronym{STACS}
\EventYear{2026}
\EventDate{March 9--September 13, 2026}
\EventLocation{Grenoble, France}
\EventLogo{}
\SeriesVolume{364}
\ArticleNo{31}
\begin{document}

\maketitle

\begin{abstract}
Let \fcal\ be a finite family of graphs.
In the \fdel\ problem, one is given a graph $G$ and an integer $k$, and
the goal is to find $k$ vertices whose deletion results in a graph with no minor from the family \fcal.
This may be regarded as a far-reaching generalization of {\sc Vertex Cover} and {\sc Feedback vertex Set}.
In their seminal work, Fomin, Lokshtanov, Misra \& Saurabh [FOCS 2012] gave a polynomial kernel for this problem when the family \fcal\ contains a planar graph.
As the size of their kernel is $g(\fcal) \cdot k^{f(\fcal)}$, a natural follow-up question was whether the dependence on \fcal\ in the exponent of $k$ can be avoided. 
The answer turned out to be negative: Giannopoulou, Jansen, Lokshtanov \& Saurabh [TALG 2017] proved that this is already inevitable
 for the special case of the \twetadel\ problem.

In this work, we show that this non-uniformity can be avoided at the expense of a small loss.
First, we present a simple 2-approximate kernelization algorithm for \twetadel\ with a kernel of size $g(\eta) \cdot k^6$.
Next, we show that 
the approximation factor can be made arbitrarily close to $1$,
if we settle for a kernelization protocol with $\Oh(1)$
calls to an oracle that solves instances of size bounded by a uniform polynomial in $k$.
We extend the above results to general \fdel, whenever \fcal\ contains a planar graph, as long as an oracle for 
\twetadel\ is available for small instances. 
Notably, all our constants are computable functions of $\fcal$ and our techniques work also when some graphs in $\fcal$ may be disconnected.

Our results rely on two novel techniques.
First, we transform so-called ``near-protrusion decompositions'' into true protrusion decompositions by sacrificing a small accuracy loss. 
Secondly, we show how to optimally compress such a decomposition with respect to general \fdel.
Using our second technique,
we also obtain linear kernels on sparse graph classes
when $\fcal$ contains a planar graph, whereas the previously known theorems required all graphs in $\fcal$ to be connected.
Specifically, we 
generalize the kernelization algorithm by Kim, Langer, Paul, Reidl, Rossmanith, Sau \& Sikdar [TALG 2015] on graph classes that exclude a topological minor.
\end{abstract}

\newpage

\pagenumbering{arabic}

\section{Introduction}\label{sec:intro}
Consider a finite family $\fcal$ of graphs.
In the \fdel\ problem we are given a graph $G$ and a positive integer $k$ and we ask whether one can remove $k$ vertices from $G$ to obtain a graph that is $\fcal$-minor-free, that is, it does not contain any graph $F \in \fcal$ as a minor.
This captures a wide variety of graph problems such as {\sc Vertex Cover}, {\sc Feedback Vertex Set}, {\sc Diamond Hitting Set}, or {\sc $d$-Path Transversal}. 
It is known that \fdel\ is fixed-parameter tractable (FPT)~\cite{DBLP:journals/talg/SauST22} for every family \fcal\ and it admits a polynomial kernel\footnote{A polynomial kernel for a parameterized problem $L$ is a polynomial-time algorithm that given an instance $(I,k)$ of $L$ returns an equivalent one $(I',k')$ such that $|I'| + k' \le \mathsf{poly}(k)$. See the book~\cite{fomin2019kernelization}.} whenever \fcal\ contains at least one planar graph~\cite{DBLP:conf/focs/FominLMS12}.
This condition is equivalent to the existence of a~constant $\eta(\fcal) \in \nn$ such that every \fcal-minor-free graph has treewidth~\cite[\S 7]{DBLP:books/sp/CyganFKLMPPS15} bounded by $\eta(\fcal)$~\cite{DBLP:journals/jct/RobertsonS86}.
A canonical problem that fits into this category is \twetadel: remove $k$ vertices from a graph to reduce its treewidth to $\eta$.
The arguably simplest case for which \fcal\ does not contain a planar graph is {\sc Vertex Planarization}; here the existence of a polynomial kernel remains a major open problem~\cite{DBLP:journals/siamcomp/JansenW25}.

What may seem like a downside of the kernelization algorithm by Fomin et al.~\cite{DBLP:conf/focs/FominLMS12} is its non-uniformity: the size of the kernel is of the form $\Oh_\fcal(k^{f(\fcal)})$ for some function $f$.
{\bf This turns out to be inevitable: \twetadel\ does not admit a kernel of size $\Oh_\eta(k^{o(\eta)})$ unless NP $\sub$ coNP/poly}~\cite{DBLP:journals/talg/GiannopoulouJLS17}.
On the positive side, uniform kernels are known for {\sc Treedepth-$\eta$-Deletion} with $\Oh_\eta(k^{6})$ vertices~\cite{DBLP:journals/talg/GiannopoulouJLS17} and for families \fcal\ of connected graphs that include $K_{2,p}$ (for any $p \in \nn$) with $\Oh_\fcal(k^{10})$ vertices~\cite{DBLP:conf/isaac/LochetS24} (cf.~\cite{DBLP:journals/siamdm/FominLMPS16}). 

Another way to obtain uniform kernelization is to restrict the input graph $G$ to some well-behaving graph class.
There are several meta-theorems that yield linear kernels for miscellaneous problems over the class of planar graphs~\cite{DBLP:journals/jacm/BodlaenderFLPST16}, for every proper minor-closed graph class~\cite{DBLP:journals/siamcomp/FominLST20}, and even for every graph class that excludes some graph as a topological minor~\cite{DBLP:journals/siamdm/GarneroPST15, kim15linear}.
These approaches suffer from a different weakness though: they require the problem in question to satisfy a certain separation condition.
{\bf This translates to a restriction that every graph in the family \fcal\ must be connected.}
{Consequently, we could not handle deletion to graph classes such as ``planar graphs of bounded face cover'' or ``bounded-treewidth graphs embeddable on a torus''.
In the second example, observe that the family of minimal minor obstructions to torus embeddability includes $K_5 + K_5$.
}

\subparagraph{Lossy kernelization.}
In order to overcome known barriers against polynomial kernels~\cite{DBLP:journals/jcss/BodlaenderDFH09, fomin2019kernelization}, researchers started to study a lossy variant of kernelization~\cite{DBLP:conf/stoc/LokshtanovPRS17}.
Intuitively, an $\alpha$-approximate kernel reduces the task of finding an approximate solution to an instance of size $\mathsf{poly}(k)$ where $k$ bounds the optimum.
Here, the approximation factor $\alpha$ measures how much is ``lost in translation'' between the two instances: a $c$-approximate solution is translated to an $(\alpha \cdot c)$-approximate solution.
Such lossy kernels have been studied for {\sc Vertex Planarization}~\cite{DBLP:journals/siamcomp/JansenW25} and for \fdel\ with a connectivity constraint on solution~\cite{DBLP:conf/wg/EibenMR22, DBLP:conf/esa/000121a}.
To capture the difficulty of an approximation task by parameter $k$, solutions larger than $k$ are treated as equally bad and we define their cost as $k+1$~\cite{DBLP:conf/stoc/LokshtanovPRS17}.

Analogously to exact kernelization, one can allow the algorithm to repeatedly call an oracle that (approximately) solves an instance of bounded size.
Assuming that the oracle processes instances of size $\mathsf{poly}(k)$, such an algorithm is called an (approximate) polynomial Turing kernelization~\cite{DBLP:conf/esa/HolsKP20, DBLP:conf/iwpec/KratschK23}.
A restricted model in which the number of oracle calls
is bounded by $\mathsf{poly}(k)$ (so, in particular,
it does not depend on the input size) has been dubbed a {\em lossy kernelization protocol}~\cite{DBLP:conf/esa/FominLL0TZ23}.
In fact, allowing just two oracle calls unlocks technical tools unavailable for currently known one-shot lossy kernels~\cite{DBLP:conf/esa/FominLL0TZ23}.
Whereas there are problems that admit an exact polynomial Turing kernel and are unlikely to admit a regular polynomial kernel~\cite{DBLP:journals/algorithmica/JansenPW19}, the current lower bound techniques do not distinguish whether the oracle is called $\mathsf{poly}(k)$ times or just once~\cite[{\S 21}]{fomin2019kernelization}.  
It is therefore an intriguing question if and how (lossy) kernelization protocols can outperform one-shot (lossy) kernels.

\subparagraph{Our results.} 
In this paper we address the two raised issues that trouble the existing kernelization algorithms for \fdel.
First, we show that a uniform polynomial kernelization for \twetadel\  is possible if we settle for 2-approximation.
{This should be compared to polynomial-time $\Oh(\log \eta)$-approximation algorithm for \twetadel~\cite{DBLP:conf/soda/GuptaLLM019} whereas $c$-approximation with an absolute constant $c$ remains elusive.}
Our result extends to \fdel\, for any family $\fcal$ containing a planar graph, modulo the fact that we still need an oracle solving \twetadel.
As the oracle problem differs from the original one, we obtain a slightly weaker result, namely a {\em lossy polynomial compression.}

\begin{restatable}[Uniform lossy kernel]{theorem}{lossykernel}
\label{thm:lossy-kernel}
Let \fcal\ be a finite family of graphs containing at least one planar graph.
  Then \fdel\ admits a 
  $2$-lossy compression 
  with $\Oh_{\mathcal{F}}(k^{5})$ vertices and of size $\Oh_{\mathcal{F}}(k^{6})$. 
   Moreover, \twetadel\  
  admits a $2$-lossy kernel 
  with $\Oh_{\mathcal{F}}(k^{5})$ vertices and of size $\Oh_{\mathcal{F}}(k^{6})$. 
\end{restatable} 

We can improve the approximation factor from 2 to $(1+\epsilon)$ for any $\epsilon >0$ by calling the oracle in several rounds.
The number of rounds, as well as the degree in the kernel size,
depend only on $\epsilon$.
Hence we obtain a uniform lossy compression protocol with approximation factor arbitrary close to 1 and with constant number of rounds.

\begin{restatable}[Uniform lossy protocol]{theorem}{lossyprotocol}
\label{thm:lossy-protocol}
Let \fcal\ be a finite family of graphs containing at least one planar graph.
  For any $\epsilon >0$, \fdel\ admits a $(1+\epsilon)$-lossy compression protocol of 
  call size $\Oh_{\mathcal{F}}(k^6 + k^{3r +1})$,
    and at most $1+r$ rounds, where $r=1+  \left \lceil  \log_{1+\epsilon} (\frac{1}{\epsilon}) \right \rceil$. 
 
   For the special case of \twetadel, there is a $(1+\epsilon)$-lossy kernelization protocol with same call size and number of rounds.
\end{restatable} 

The proofs of Theorems \ref{thm:lossy-kernel} and \ref{thm:lossy-protocol} rely on ``protrusion techniques''~\cite[\S 16]{fomin2019kernelization}.
An~{\em $r$-protrusion} is a vertex subset $X \sub V(G)$ such that (i) treewidth of $G[X]$ is at most $r$ and (ii) $X$ contains at most $r$ vertices with a neighbor outside $X$.
A graph $G$ admits an {\em $(\alpha, \beta)$-protrusion decomposition} if there is a vertex subset $P_0 \sub V(G)$ of size $\le \alpha$ so that $G-P_0$ can be covered by $\alpha$ many $\beta$-protrusions (see \Cref{sec:prelims} for formal definitions).
In the proofs we transform so-called {\em near-protrusion decompositions}~\cite{DBLP:conf/focs/FominLMS12} into true protrusion decompositions with bounded parameters $\alpha, \beta$ by sacrificing small accuracy loss.
The known techniques can compress such structures in a uniform fashion as long as all graphs in the family $\fcal$ are connected. 
To tackle the remaining cases, we prove the following lemma.
Here, $\opt_\fcal(G)$ denotes the minimum size of a solution to \fdel\ in $G$, i.e., an $\fcal$-deletion set.

\begin{lemma}[Handling disconnected forbidden minors]\label{thm:dichotomy-intro}
    Let $\cal F$ be a finite family of graphs. 
    There is an algorithm that, given a graph $G$
    with an $(\alpha, \beta)$-protrusion decomposition and integer $k$, runs in time $\Oh_{\fcal,\beta}(|V(G)|)$ and
    outputs a graph $G'$ with $\Oh_{\fcal,\beta}(\alpha + k)$ vertices
    such that $\min(\opt_\fcal(G), k+1) = \min(\opt_\fcal(G'), k+1)$.  
\end{lemma}

We remark that the factor hidden in the $\Oh(\cdot)$-notation is a computable function of $(\fcal,\beta)$.
See \Cref{lem:dichotomy:compress} for a full statement tailored for lossy kernelization.
As a corollary of \Cref{thm:dichotomy-intro} we obtain linear kernels on graph classes where one can compute an $(\Oh(k), \Oh(1))$-protrusion decomposition.
The most general cases where such a construction is known are classes with excluded topological minor~\cite{kim15linear}.
Whereas the known linear kernels for \fdel\ require all the graphs in $\fcal$ to be connected~\cite{DBLP:journals/jacm/BodlaenderFLPST16, DBLP:journals/siamcomp/FominLST20, DBLP:journals/siamdm/GarneroPST15,  kim15linear}, we are able to
 drop this assumption.

\begin{restatable}[Linear kernel on sparse graphs]{theorem}{thmSparse}
\label{thm:sparse}
    Let $H$ be a graph and $\mathcal{F}$ be a finite family of graphs containing at least one planar graph. 
    Then \fdel\ admits a linear kernel on $H$-topological-minor-free graphs.
\end{restatable}

\subparagraph{Organization of the paper.}
We begin with an informal exposition of our main technical ideas in \Cref{sec:overview}.
The most of the remaining space is devoted to \Cref{thm:dichotomy-intro} (\Cref{sec:disconnected}) which is too technical to be covered in the overview.
The proofs of lemmas marked with $(\bigstar)$ are postponed to the appendix, as well as the proofs of the main theorems.

\section{Preliminaries}\label{sec:prelims}
For any object $F$, the notation $\Oh_F(\cdot)$ hides factors depending on $F$.
For $p \in \mathbb{N}$ we denote $[p] = [1,p] =\{1, \ldots,p\}$.
%By the {\em size} of a graph, we mean the number of vertices in it.

\begin{definition}\label{def:prelim:lca}
Let $T$ be a rooted tree and $S \subseteq V(T)$ be a set of vertices in $T$. 
We define the least common ancestor of (not necessarily distinct) $u, v \in V(T)$, denoted as {$\mathsf{LCA}(u, v)$}, to be the deepest node $x$ which is an ancestor of both $u$ and $v$.
The LCA closure of $S$ is the set
\[
\overline{\mathsf{LCA}}(S) = \{\mathsf{LCA}(u, v): u, v \in S\}.
\]
\end{definition}

\begin{lemma}[{\cite[Lem. 9.26, 9.27, 9.28]{fomin2019kernelization}}]\label{lem:prelim:lca}
Let $T$ be a rooted tree, $S \subseteq V(T)$ be non-empty, and $L = \overline{\mathsf{LCA}}(S)$. {All of the following hold.}
\begin{enumerate}
    \item Each connected component $C$ of $T - L$ satisfies $|N_T(C)| \le 2$.
    \item $|L| \le 2\cdot |S| - 1$.
    \item $\overline{\mathsf{LCA}}(L) = L$.
\end{enumerate}
\end{lemma}

\subparagraph*{Protrusions.} 
For any $\beta \in \mathbb{N}$,  
a \emph{$\beta$-protrusion} in a graph $G$ is a vertex set $X \subseteq V(G)$ 
such that the treewidth of $G[X]$ is at most $\beta$ and 
$|\partial_G(X)| \leq \beta$, where $\partial_G(X) = N_G(V(G) \sm X)$.
When $G$ is clear from the context, we drop the subscript.

\begin{definition}
For $\alpha,\beta \in \mathbb{N}$, 
an $(\alpha,\beta)$-protrusion decomposition of a graph $G$ is a partition of the vertex set $V(G)=(P_0, P_1, \ldots, P_{\ell})$ such that:

\begin{enumerate}
    \item $\max\{|P_0|,\, \ell\} \leq \alpha$, and
    \item for each $i \in [1,\ell]$, $N[P_i]$ is a $\beta$-protrusion, and
    \item for each $i \in [1, \ell]$, $N(P_i) \subseteq P_0$.
\end{enumerate}
We refer to $P_0$ as the {\em root bag} of the decomposition.
{A protrusion decomposition is called {\em nice} if for each $i \in [1,\ell]$ it holds that $|N(P_i)| \le \beta$.}
\end{definition}

\begin{lemma}\label{lem:prelims:protrusion-neighborhood}
    An $(\alpha,\beta)$-protrusion decomposition $(P_0, P_1, \ldots, P_{\ell})$ of $G$ can be transformed in linear time into a nice $(\alpha,\beta)$-protrusion decomposition $(P'_0, P'_1, \ldots, P'_{\ell})$ such that for each $i \in [1, \ell]$ it holds that $P_i \subseteq P'_i$.
\end{lemma}
\begin{proof}
    Fix $i \in [1, \ell]$. By definition, we have that $|\partial(N[P_i])| = |N(V(G) \sm N[P_i])| \le \beta$. Suppose that $v \in N(P_i) \sm \partial(N[P_i])$.
    Then $N(v) \sub N[P_i]$  
    and moving $v$ from $P_0$ to $P_i$ does not alter the set $N[P_i]$, hence $N[P_i] \cup \{v\}$
    remains a $\beta$-protrusion.
    After performing such an operation exhaustively for each $i \in [1, \ell]$ and $v \in N(P_i) \sm \partial(N[P_i])$, we obtain a nice $(\alpha,\beta)$-protrusion decomposition.
\end{proof}

\subparagraph*{Boundaried graphs.}
A $t$-boundaried graph is a triple ${\bf H} = (H,B,\lambda)$ where $H$ is a graph, $B \sub V(H)$ is of size $t$ 
and $\lambda \colon [t] \to B$ is a bijection representing a labeling function. 
We refer to the set $B$ as $\partial {\bf H}$ and call it the {\em boundary} of ${\bf H}$. 
By $V({\bf H})$ we refer to $V(H)$.
Given two $t$-boundaried graphs ${\bf H_1} =(H_1,B_1,\lambda_1)$ and ${\bf H_2} =(H_2,B_2,\lambda_2)$,
${\bf H_1}   \oplus   {\bf H_2}$ denotes the graph $H$ obtained by first taking the disjoint union of $H_1$ and $H_2$ and then identifying each vertex $u_1$ in $B_1$ with $u_2 \in B_2$ such that $\lambda_1^{-1}(u_1)= \lambda_2^{-1}(u_2)$.

For any $h \in \mathbb{N}$
and a graph $G$, {\em $h$-{folio}$(G)$} is the set of all graphs on at most $h$ vertices which appear in $G$ as a minor.

\begin{restatable}[$\bigstar$]{lemma}{protrusionreplacer}\label{lem:prelim:replace}
    There is a computable function $\gamma \colon \nn \times \nn \times \nn \to \nn$ so that the following holds.
    Let  $h, r, d \in \nn$ and $\fcal$ be a family of graphs where each graph has at most $h$ vertices.
    Next, for a graph $G$ and an $r$-protrusion $X$, let us denote ${\bf H} = (G[X], \partial_G(X), \lambda)$, where $\lambda$ is some labeling function, and ${\bf H^*}=(G[N[V(G) \setminus X]] ,  \partial_G(X), \lambda)$.
    Note that $G= {\bf H} \oplus {\bf H^*}$.
    
    Then there is an algorithm that, given $G$ and $X$, 
    runs in time $\Oh_{h,r,d}(|X|)$ and returns a graph $G'$ with the following properties.

    \begin{enumerate}
        \item $G'$ is obtained from $G$ by replacing the boundaried graph ${\bf H}$ with a $|\partial_G(X)|$-boundaried graph ${\bf \widehat H}$, that is $G'= {\bf \widehat{H}} \oplus {\bf H^*}$
        such that:  
        \begin{itemize}
            \item $h$-{\em folio}$({\bf H} - \partial {\bf H}) = h$-{\em folio}$({\bf {\widehat  H} } - \partial {\bf {\widehat  H} } ) $  
            \item $|V({\bf \widehat H })| \le \gamma(h,r,d)$.
        \end{itemize}
        \item Given an \fcal-deletion set $S'$ in $G'$ such that $|S' \cap V({\bf \widehat H })| \le d$, one can in polynomial time find an \fcal-deletion set $S$ in $G$ of size at most $|S'|$. 
        And vice versa, an \fcal-deletion set $S$ in $G$ such that $|S \cap V({\bf H})| \leq d$ can be lifted to an \fcal-deletion set $S'$ of size at most $|S|$ in $G'$.     
    \end{enumerate}
    {Additionally, suppose that $Q$ is a fixed graph and $G$ is $Q$-topological-minor free.
    Then the algorithm outputs a graph $G'$ that is also $Q$-topological-minor free in time $\Oh_{h,r,d,Q}(|X|)$.}

\end{restatable}

The proof of Lemma~\ref{lem:prelim:replace} appears in Section~\ref{sec:computable-protrusion-repalcer} for completeness and follows the proof of~\cite[Theorem~$1.1$]{DBLP:journals/siamdm/GarneroPST15}. 
Note that the key highlight of Lemma~\ref{lem:prelim:replace} is that the function $\gamma$ is computable and the family $\mathcal{F}$ is arbitrary (and not necessarily of connected graphs).

\section{Technical overview}
\label{sec:overview}

We sketch the main ideas in Theorems \ref{thm:lossy-kernel} and \ref{thm:lossy-protocol} for the simplest case of \twetadel.
The starting point is a {near-protrusion decomposition} from~\cite{DBLP:conf/focs/FominLMS12}:
one can assume that graph $G$ is equipped with disjoint sets $X, Z$ so that $|X| = \Oh_\eta(k)$, $|Z| =\Oh_\eta(k^3)$, $\tw(G-X) \le \eta$, each component $C$ of $G - (X\cup Z)$ has $\Oh_\eta(1)$
neighbors in $Z$ and each $u,v \in N(C) \cap X$ are highly connected.
The last condition implies that for any solution $S$ of size $\le k$, if $\{u,v\} \cap S = \emptyset$ then $u,v$ must share a bag in any tree decomposition of $G-S$ of width $\le \eta$.
Therefore, it is safe to insert the edge $uv$ to $G$.
We can thus assume that $N(C) \cap X$ is a clique.
If this clique has size $\ge 2(\eta+1)$ then we know that any solution must remove at least half of its vertices.
And so we can remove the entire clique by paying 2 in the approximation factor.
After applying such a reduction rule, each component $C$ of $G - (X\cup Z)$ has $\Oh_\eta(1)$  neighbors in total.
This implies that $N[C]$ forms an {$\Oh_\eta(1)$-protrusion}, having both its treewidth and boundary bounded. 
The framework of protrusion replacement~\cite{DBLP:journals/jacm/BodlaenderFLPST16, DBLP:journals/siamdm/GarneroPST15} allows us to {replace} $G[N[C]]$ by a subgraph of constant size that exhibits the same behavior with respect to the problem in question.
The caveat is that the number of such protrusions may be as large as $k^{\Omega(\eta)}$ which is prohibitive for constructing a uniform kernel. 

\subparagraph*{Uniform lossy kernel.}
We would like to transform the current decomposition into a $(\Oh_\eta(k^c), \Oh_\eta(1))$-protrusion decomposition, where $c$ does not depend on $\eta$.
We call a component  $C$ of $G - (X\cup Z)$ {\em simplicial} if $N(C)$ forms a clique.
Because we can insert edges between highly connected vertices in $X \cup Z$, each non-edge may appear only in a bounded number of sets $N(C)$.
We use this argument to bound the number of non-simplicial components by $\Oh_\eta(k^5)$.
Let us neglect the simplicial components for a moment and let $G'$ denote the graph without them.
Then $G'$ admits a protrusion decomposition of desired form and
we can compress $G'$ to $G''$ on $\Oh_\eta(k^5)$ vertices.
This constitutes the output of the kernelization algorithm.

It remains to provide a mechanism to lift a solution $S''$ from $G''$ to $G$. 
First, we note that a solution $S''$ can be lifted to a solution $S'$ of $G'$ using standard tools. 
The interesting step is lifting $S'$ to a solution $S$ in $G$. This is the part where the lifting step of the algorithm does a nontrivial (yet, polynomial-time) computation exploiting a treewidth bound.
Consider a tree decomposition $\cal T$ of $G'-S'$ of width $\eta$ and a simplicial component  $C$ of $G - (X\cup Z)$.
Since $N(C) \sm S'$ is a clique, it must be covered by some bag in $\cal T$.
Therefore, we can plug $G[N[C]]$ into $\cal T$ by merging a tree decomposition of $G[N[C]]$ of width $2\eta +1 $ and the tree decomposition $\mathcal{T}$ alongside $N(C) \sm S'$.
This argument bounds the treewidth of $G-S'$ by $2\eta+1$
and on such a graph we can solve \twetadel\ exactly in polynomial time.
We output the union of $S'$ and the optimal solution in $G-S'$, resulting in 2-approximation.

\subparagraph*{Uniform lossy protocol.}
Improving the approximation factor below $2$ requires a different approach. 
We can still start with the sets $X \cup Z$ as described in the beginning of this section, and make each connected component $C$ of $G-(X \cup Z)$ a protrusion, this time by removing a clique in the neighbourhood of $C$ whose size is at least $(\eta+1) (1+\epsilon) / \epsilon$.
The main challenge is to bound the number of simplicial protrusions.
Let $Y_0 = X \cup Z$.
Since we now aim at a lossy kernelization protocol, we can send $G[Y_0]$ to the oracle which outputs $Y_1 \sub Y_0$ for which $\tw(G[Y_0 \sm Y_1]) \le \eta$.
Assume for simplicity that the oracle always returns an optimal solution.
Then $|Y_1| = \opt(G[Y_0]) \le \opt(G)$.
And once again, we ask the oracle for a solution $Y_2$ in $G[Y_1]$.

Suppose first that $|Y_2| > (1 - \eps) \cdot |Y_1|$.
Observe that any optimal solution $S$ in $G$ satisfies $|S \cap Y_1| \ge \opt(G[Y_1])$ so $\opt(G-Y_1) \le \opt(G) -  (1 - \eps) \cdot |Y_1|$.
A~crucial observation is that removing $Y_1$ from $G$ 
allows us to construct an $(\Oh_\eta(k^7), \Oh_\eta(1))$-protrusion decomposition.
For each simplicial component $C$, the set $N(C) \sm Y_1$ forms a clique in the graph $G[Y_0 \sm Y_1]$ of treewidth $\le \eta$.
But the number of cliques in a bounded-treewidth graph is linear so we can group the simplicial components into $\Oh_\eta(|Y_0 \sm Y_1|)$ protrusions according to their neighborhoods in $Y_0 \sm Y_1$.
Together with $\Oh_\eta(k^5)$ non-simplicial components, this allows us to compress the graph $G-Y_1$ into size $\Oh_\eta(k^5)$. 
Let $S_1$ be the solution in $G-Y_1$ computed with the help of the oracle.
Assuming $|S_1| = \opt(G-Y_1)$ we get $|S_1| + |Y_1| \le  \opt(G) -  (1 - \eps) \cdot |Y_1| + |Y_1| = \opt(G) + \epsilon \cdot |Y_1| \le (1+ \epsilon)\cdot \opt(G)$.

Suppose now that  $|Y_2| \le (1 - \eps) \cdot |Y_1|$.
We continue asking the oracle for solution $Y_{i+1}$ in $G[Y_i]$.
Let $j$ be the first iteration for which $|Y_{j+1}| > (1 - \eps) \cdot |Y_{j}|$.
We repeat the same trick but this time the set $Y_0 \sm Y_{j}$
is partitioned into $j$ subsets $(Y_0 \sm Y_1), (Y_1 \sm Y_2), \dots, (Y_{j-1} \sm Y_j)$, each inducing a graph of treewidth $\le \eta$.
The number of cliques in $G[Y_0 \sm Y_j]$ can be bounded 
by $f(\eta,j)\cdot |Y_0|^j$, 
allowing us to construct an $(\Oh_{\eta,j}(k^{\Oh(j) + 5}), \Oh_{\eta,j}(1))$-protrusion decomposition.

This procedure yields a uniform kernelization protocol as long as $j$ can be bounded.
But suppose we keep getting outcome $|Y_{j+1}| \le (1 - \eps) \cdot |Y_{j}|$.
Then for  $r=1+  \left \lceil  \log_{1-\epsilon} ({\epsilon}) \right \rceil$ we obtain $|Y_r| \le (1-\eps)^{r-1} \cdot |Y_1| = \eps \cdot |Y_1| \le \eps \cdot \opt(G)$.
In this case we can again use the oracle to compute a solution in $G-Y_r$, bounding the number of protrusions by $|Y_0|^r$, and merge it with $Y_r$, achieving $(1+\eps)$-approximation.
See \Cref{fig:lossy-protocol-oracle} on page \pageref{fig:lossy-protocol-oracle} for an illustration.

In order to handle \fdel, we need to be more careful because inserting an edge between highly connected vertices may no longer be safe.
To this end, we introduce an auxiliary graph which collects these inserted edges and exploit the fact that the optimum to \fdel\ can be lower bounded by the optimum to \twetadel, for some $\eta = \eta(\fcal)$.
We need to ask the oracle for solutions to  \twetadel\ over the auxiliary graph as well, and this is the reason why in general we end up with compression protocol, rather than kernelization protocol. 
Whereas the standard protrusion replacement technique suffices to finish the construction when all the graphs in $\fcal$ are connected, dealing with the general case requires one more tool.

\subparagraph*{Handling disconnected forbidden minors.}
There are two known ways to compress protrusions in the presence of disconnected graphs in $\fcal$.
First, one can consider an annotated version of \fdel, where some vertices are marked as undeletable~\cite{DBLP:journals/jacm/BodlaenderFLPST16}, but this requires the oracle to solve a significantly harder task.
Secondly, one can take advantage of a certain {\em well-quasi-order} over boundaried graphs~\cite{DBLP:conf/focs/FominLMS12} but this involves a non-constructive argument and results in factors that a priori may be an uncomputable function of $\fcal$.
Besides, both approaches are insufficient to construct a linear kernel in~\Cref{thm:sparse}.

To visualize the issue, 
consider $F = K_4 + C_5$ and
an $(\alpha,\beta)$-protrusion decomposition $(P_0,P_1,\dots,P_\ell)$ of a graph $G$.
It may be the case that $G[P_i]$ is $K_4$-free for some $i\in [1,\ell]$ but it contains a lot of minor models of $C_5$.
Depending on whether a solution to {\sc $F$-Deletion} hits all the copies of $K_4$ outside $P$, it may need to remove either none or a very large number of vertices in $P$.
In fact, such a problem does not admit {\em finite integer index}~\cite{kim15linear}, a property crucial for protrusion replacement~\cite{DBLP:journals/jacm/BodlaenderFLPST16, DBLP:journals/siamdm/GarneroPST15}.

Observe however that if $G[P_i]$ contains a disjoint packing of $k+1$ minor models of $C_5$, then we already know that no solution of size $\le k$ can make the graph $C_5$-minor-free, so it should ``focus'' on hitting the $K_4$-minors.
There is a caveat though: it might be the case that the only minor model of $K_4$ in $G$ intersects each model of $C_5$ and so $K_4 + C_5$ already does not appear as a minor.
In order to circumvent such error-prone corner cases, we refine an $(\alpha,\beta)$-protrusion decomposition to one with {\em dichotomy property}.
It states that for every connected graph $F$ on at most $h$ vertices (where $h$ is the maximum graph size in $\fcal$) either $F$ does not appear as a minor in $G-P_0$ or $F$ has a large ``private'' collection of protrusions in which it appears. 
This property implies that any optimal solution $S$ can use only $\Oh_{\fcal,\beta}(1)$ vertices from each protrusion, mimicking the missing separation property.
Intuitively, the aforementioned collections justify that any solution $S$ of size $\le k$ must focus on hitting minors that do not appear in $G-P_0$, and so it does not make sense to remove too many vertices from a single protrusion.
Consequently, the constructive protrusion replacement by Garnero et al.~\cite{DBLP:journals/siamdm/GarneroPST15} can be adapted to compress protrusions in a decomposition with dichotomy property.

To construct a decomposition with dichotomy property 
we generalize the packing/covering duality for connected minor models in bounded-treewidth graphs~\cite{DBLP:journals/dam/RaymondT17}.
In our case however we do not pack models of a single graph $F$ but we consider a family of connected graphs $\fcal'$ and in each step of a greedy procedure we need to choose one graph $F \in \fcal'$ to be packed, without spoiling the invariants for the remaining graphs from $\fcal'$.
The details are technical and we omit them here.

\section{Protrusion replacement for disconnected forbidden minors}\label{sec:disconnected}

In this section we collect ingredients to prove \Cref{lem:dichotomy:compress} which is an extension of \Cref{thm:dichotomy-intro}.
We begin with two crucial definitions.
We write $2^{[1,\ell]}$ to denote the power set of $[1,\ell]$.

\begin{definition}[Minor packing] 
    Let $c \in \nn$ and $\fcal$ be a finite family of graphs.
    We say that a protrusion decomposition  $(P_0,P_1,\dots,P_\ell)$ of $G$ admits an $(\fcal,c)$-minor packing if there is a mapping $I \colon \fcal \to 2^{[1,\ell]}$ such that:
    \begin{enumerate}
        \item for each $F \in \fcal$, $|I(F)| =c$,
        \item for each $F \in \fcal,\,i \in I(F),$ it holds that $F \le_m G[P_i]$,
        \item for each distinct $F_1,F_2 \in \fcal$, the sets $I({F_1}), I({F_2})$ are disjoint.
    \end{enumerate}
\end{definition}

\begin{definition}[Dichotomy property]
    Let $k,h \in \mathbb{N}$ and ${\cal H}^{\mathsf{conn}}_h$ denote the family of all connected graphs on at most $h$ vertices.
    
    We say that a protrusion decomposition $(P_0,P_1,\dots,P_\ell)$ of graph $G$ has $(k,h)$-dichotomy property if it admits an $(\fcal,k+h)$-minor packing, where $\fcal \sub {\cal H}^{\mathsf{conn}}_h$ is the family of connected graphs that belong to  $h$-{\em folio}$(G-P_0)$.
\end{definition}

We prove that any protrusion decomposition can be modified to satisfy $(k,h)$-dichotomy property.
After initializing $\fcal = {\cal H}^{\mathsf{conn}}_h$,
we consider each graph from $\fcal$ and as long as $G-P_0$ contains an $F$-deletion set $S_F$ of size $\Oh_{h,\beta}(k)$ for some $F \in \fcal$, we insert $S_F$ into $P_0$ and remove $F$ from $\fcal$.
More precisely, we first compute the LCA-closure (see \Cref{app:prelims}) for $S_F$ in a tree decomposition of $G-P_0$ to preserve the invariants of a protrusion decomposition.
When we cannot proceed, each remaining graph  $F \in \fcal$ must have a large packing of disjoint minor models in $G-P_0$.
Note that these models can intersect for different $F$.
We give a greedy procedure that scans a tree decomposition of $G-P_0$ bottom-up  and marks bags that allow us to separate a new protrusion containing a model of $F$ (see  \Cref{fig:dichotomy:multi-packing} on page \pageref{fig:dichotomy:multi-packing}).
We also care to keep the set of marked bags closed under LCA, so that we end up with a protrusion decomposition of $G-P_0$ with an $(\fcal, k+h)$-minor packing.
Finally, we merge it with the protrusion decomposition of $G$.

\begin{restatable}{lemma}{lemDichoConstruct}
\label{lem:dichotomy:construct}
    Suppose $G$ admits an $(\alpha, \beta)$-protrusion decomposition  $(P_0,P_1,\dots,P_\ell)$.
    Then it also admits a nice $(\alpha + \Oh_{h,\beta}(k),\, \Oh_{h,\beta}(1))$-protrusion decomposition with $(k,h)$-dichotomy property. 
    Furthermore, this new decomposition, together with the corresponding minor packing, can be constructed in time $\Oh_{h,\beta}(|V(G)|)$ given the old one.
\end{restatable}

\begin{figure}[t]
    \centering
    \includegraphics[width=\linewidth]{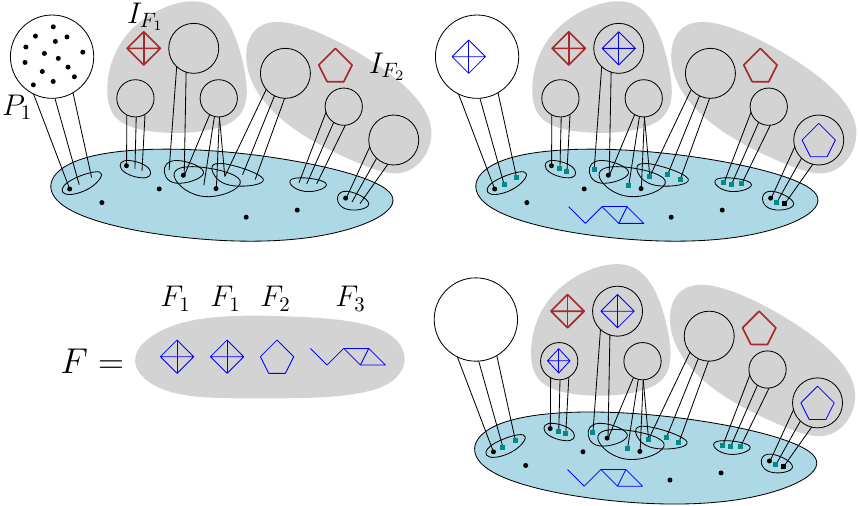}
    \caption{Illustration to \Cref{lem:dichotomy:bounded} for $F$ being a disjoint union $F_1 + F_1 + F_2 + F_3$ and $\fcal = \{F\}$.
    Top left: An $\fcal$-deletion set $S$ is marked with black discs.
    The gray areas correspond to those $P_i$ which belong to $I({F_1}), I({F_2})$.
    Their neighborhoods in $P_0$ are sketched as well.
    Each $P_i$ in the corresponding family contains a minor model of the graph drawn in brown and is disjoint from $S$.
    Top right: $\widehat S$ is obtained from $S$ by removing $S \cap P_1$ and adding  the vertices marked with squares.
    A potential minor model $\Phi$ of $F$ in $G-\widehat S$ is drawn in blue.
    It places one copy of $F_1$ inside $P_1$.
    Bottom right: A new model $\Phi'$ of $F$ is obtained from $\Phi$ by moving the image of $F_1$ to an unused component in  $I({F_1})$.
    But this model is also present in $G-S$ which leads to a contradiction.
    }
    \label{fig:dichotomy:bounded}
\end{figure}

The proof of Lemma~\ref{lem:dichotomy:construct} is a technical generalization of known ideas and we 
defer it to Section~\ref{sec:dichotomy-prop} to keep the focus of the current section.
We now explain how dichotomy property ensures that an optimal $\fcal$-deletion set can remove only a bounded number of vertices from each protrusion.

\begin{lemma}\label{lem:dichotomy:bounded}
    Suppose $G$ admits a nice $(\alpha, \beta)$-protrusion decomposition $(P_0,P_1,\dots,P_\ell)$ with $(k,h)$-dichotomy property and let $\cal F$ be a family of graphs of maximum size $h$. 
    Next, suppose that $G$ has an $\cal F$-deletion set $S$ of size  $\le k$.
    Then $G$ has an $\cal F$-deletion set $\widehat S$ with the property that for each $i \in [1, \ell]$ we have $|\widehat S \cap P_i| \le \beta\cdot(h\cdot |{\cal H}^{\mathsf{conn}}_h| + 1)$ and $|\widehat S| \le |S|$.
    
    Furthermore, one can compute $\widehat S$, given $G$, $S$ and $(P_0,P_1,\dots,P_\ell)$, in polynomial time.
     Moreover, every minimum $\cal F$-deletion set in $G$ has the described property.
\end{lemma}
\begin{proof}
    Let ${\cal H}_G \sub {\cal H}^{\mathsf{conn}}_h$ be the family of those connected graphs on at most $h$ vertices that appear as minors in $G - P_0$.
    Due to $(k,h)$-dichotomy property, for each $H \in {\cal H}_G$ there is a collection $I(H) \sub [1,\ell]$ of size $k+h$, such that for $i \in I(H)$ we have $H \le_m G[P_i]$ and the sets $I(H)$ are disjoint for distinct graphs $H$ from  ${\cal H}_G$.
    Since $|S| \le k$, for some $h$ indices $i \in I(H)$ the set $P_i$ is disjoint from $S$. 
    Let us denote this subset as $J_H \sub I(H)$.
    Consequently, for each $H \in {\cal H}_G$ we have $|J_H| = h$, the sets $J_H$ are pairwise disjoint, and for each $i \in J_H$ it holds that $H \le_m G[P_i]$ and $P_i \cap S = \emptyset$, 
    Note that $\sum_{H \in {\cal H}_G} |J_H| \le h \cdot |{\cal H}^{\mathsf{conn}}_h|$.

    Suppose that there exists $j \in [1,\ell]$ for which $|S \cap P_j| > \beta\cdot(h\cdot |{\cal H}^{\mathsf{conn}}_h| + 1)$.
    W.l.o.g. we will assume that $j=1$.
    Let $\widehat S$ be obtained from $S$ by (1) removing $S \cap P_1$, (2) inserting $N(P_1)$, (3) inserting $N(P_i)$ for every $i \in J_H$, $H \in {\cal H}_G$ (see \Cref{fig:dichotomy:bounded}).
    Note that $|\widehat S| < |S|$ because the given protrusion decomposition is nice.

    We will now prove that $\widehat S$ is also an \fcal-deletion set.
    Suppose not, and let $F \in \fcal$ be a graph appearing in $G - \widehat S$ as a minor.
    Let $F_1,\dots,F_c$, $c \le h$, be the connected components of $F$.
    Note that some of them may be isomorphic.
    Consider a minor model 
    $\Phi$ of $F$ in $G - \widehat S$ for which the number of $F_j$ appearing in $P_1$ is minimized.
    More precisely, we minimize the number $c'$ of indices $j \in [c]$
    for which $\Phi(V(F_j)) \sub P_1$. 
    Observe that $P_1$ is a union of connected components of   $G - \widehat S$ so if $\Phi(V(F_j))$ intersects $P_1$ then it must be fully contained in $P_1$.
    If $c'=0$ then $F$ is a minor of  $G - (P_1 \cup \widehat S)$.
    But this is an induced subgraph of $G-S$, which is \fcal-minor-free by assumption, so this is impossible.

    Consequently, the  minor model $\Phi$ places $c' \ge 1$ connected components of $F$ inside $P_1$.
    Assume w.l.o.g. that this is the case for $F_1$.
    Aiming at contradiction, we want to modify $\Phi$ to place $F_1$ outside $P_1$, without affecting the models of remaining $F_j$, $j > 1$.
    Since $F_1 \le_m G[P_1]$ it follows that $F_1 \in {\cal H}_G$.
    Recall that for each $i \in J_{F_1}$ it holds that $\widehat S \cap P_i = \emptyset$ (because  $S \cap P_i = \emptyset$) and $N(P_i) \sub \widehat S$.
    Since each $F_j$ is connected, 
    $\Phi(F_j)$ can intersect at most one $P_i$ for  $i \in J_{F_1}$.
    Next, $|J_{F_1}| = h$ and $c \le h$, so there is some $i \in J_{F_1}$ for which $P_i$ is disjoint from all $\Phi(F_j)$, $j > 1$.
    Therefore, we can move the model of $F_1$ from $P_1$ to $P_i$, obtaining a model $\Phi'$ of $F$ in $G-\widehat S$ with $c'-1$ components contained in $P_1$.
    This contradicts the choice of the model $\Phi$ and, consequently, contradicts the assumption that $F \le_m G - \widehat S$.
    Hence $\widehat S$ is a valid \fcal-deletion set. 

    We can repeat this process to reduce the intersection of  $\widehat S$ with each $P_i$ because the described modification to $S$ never adds new vertices from any $P_i$.
    A single refinement step can be implemented in polynomial time using any polynomial-time algorithm for minor testing, for example~\cite{DBLP:journals/jct/RobertsonS95b,DBLP:conf/focs/KorhonenPS24}, which also yields an algorithm to construct $\widehat S$. 
    Finally, observe that if $G$ has an \fcal-deletion set of size at most $k$,
    then every minimum $\fcal$-deletion set $S$ in $G$ must satisfy $|S \cap P_i| \le \beta \cdot (h \cdot |{\cal H}^{\mathsf{conn}}_h| + 1)$ as otherwise we could find a smaller $\fcal$-deletion set.
\end{proof}

We can now combine \Cref{lem:dichotomy:bounded} with \Cref{lem:prelim:replace} to compress each $P_i$ in a protrusion decomposition to a constant size.
It is crucial that \Cref{lem:prelim:replace} preserves the $h$-folio of $G[P_i]$ therefore the dichotomy property is maintained after the replacement.
The last point of the lemma statement refers to the solution cost capped at $k+1$ (see \Cref{sec:prelims-lossy}), which will be convenient in the design of a lossy kernelization protocol.

\begin{restatable}{lemma}{lemDichoCompress}
\label{lem:dichotomy:compress}
    Let $k,h \in \nn$ and $\cal F$ be a family of graphs of maximum size $h$.
    There is an algorithm that, given a graph $G$
    with an $(\alpha, \beta)$-protrusion decomposition, runs in time $\Oh_{h,\beta}(|V(G)|)$ 
    and 
    returns a graph $G'$ on $\Oh_{h,\beta}(\alpha + k)$ vertices
    such that $\min(\opt_\fcal(G), k+1) = \min(\opt_\fcal(G'), k'+1)$.  
    
    Furthermore, given an \fcal-deletion set $S'$ in $G'$, one can in polynomial time lift it to an \fcal-deletion set $S$ in $G$, satisfying
    \begin{align*}
    \frac{\val^{\fdelsmall}_{(G,k)}(S)}{\opt_{\fdelsmall}((G,k))} \leq \frac{\val^{\fdelsmall}_{(G',k)}(S')}{\opt_{\fdelsmall}((G',k'))}.
    \end{align*}
\end{restatable}

\begin{proof}
    We apply \Cref{lem:dichotomy:construct} to construct a nice $(\widehat \alpha, \widehat \beta)$-protrusion decomposition $(P_0,P_1,\dots,P_\ell)$ with the $(k,h)$-dichotomy property, where $\widehat \alpha = \alpha + \Oh_{h,\beta}(k)$ and $\widehat\beta = \Oh_{h,\beta}(1)$.
    Let $d := \widehat\beta\cdot(h\cdot |{\cal H}^{\mathsf{conn}}_h| + 1)$.
    We will iteratively replace each protrusion induced by $N_G[P_i]$ with one of constant size 
    using \Cref{lem:prelim:replace}.
    Below we describe this process for $P_1$.

    We fix some labeling $\lambda$ of $\partial_G(N_G[P_1]))$ and
    replace the boundaried graph ${\bf H} = (G[N[P_1]], \allowbreak N(P_1), \allowbreak \lambda)$ with a boundaried graph ${\bf \widehat H}$ of size $\Oh_{h,\widehat\beta,d}(1) \in \Oh_{h,\beta}(1)$.
    Denote the new graph $G'$ and let $P'_1$ be the set of vertices put in place of $P_1$.
   
    \begin{claim}\label{claim:dichotomy:forward}
    Suppose that  $\opt_\fcal(G) \le k$. Then $\opt_\fcal(G') \le \opt_\fcal(G)$.
    \end{claim}
    \begin{proof}
        By \Cref{lem:dichotomy:bounded} there exists an optimal \fcal-deletion set $S$ in $G$ such that $|S \cap P_1| \le d$.
        \Cref{lem:prelim:replace} guarantees that such $S$ can be translated to  an  $\fcal$-deletion set in $G'$ of size at most~$|S|$.
    \end{proof}

    \begin{claim}\label{claim:dichotomy:preserve}
        The protrusion decomposition $(P_0,P'_1,P_2,\dots,P_\ell)$ of $G'$ admits $(k,h)$-dichotomy property.
    \end{claim}
    \begin{proof}
        \Cref{lem:prelim:replace} ensures that
         $h$-folio$({\bf H-\partial H}) = h$-folio$({\bf \widehat H}-\partial{\bf  \widehat H})$ which can be rewritten as  $h$-folio$(G[P_1]) = h$-folio$(G'[P'_1])$.
        If $P_1$ contributed a minor model to the minor packing corresponding to $(k,h)$-dichotomy property, it can be replaced by a minor model of the same graph in $G'[P'_1]$
    \end{proof}

    \begin{claim}\label{claim:dichotomy:backward}
        Suppose that  $\opt_\fcal(G') \le k$. Then $\opt_\fcal(G) \le \opt_\fcal(G')$.
        Furthermore, given an \fcal-deletion set $S'$ in $G'$ of size $\le k$, one can in polynomial time lift it to an \fcal-deletion set $S$ in $G$ of size at most $|S'|$.
    \end{claim}
    \begin{proof}
        The previous claim allows us to apply
        \Cref{lem:dichotomy:bounded} to graph $G'$ and $S'$.
        Given $S'$ we can in polynomial time find an \fcal-deletion set $\widehat S$ in $G'$ of size $\le |S'|$ for which $|\widehat S \cap P'_1| \le d$.
        By applying \Cref{lem:prelim:replace} we can lift $\widehat S$ to a solution in $G$ of size at most $|\widehat S| \le |S'|$.
        By taking $S'$ to be an optimal $\fcal$-deletion set in $G'$, we infer that 
        $\opt_\fcal(G) \le \opt_\fcal(G')$.
    \end{proof}

    Claims \ref{claim:dichotomy:forward} and \ref{claim:dichotomy:backward} imply that $\min(\opt_\fcal(G), k+1) = \min(\opt_\fcal(G'), k+1)$
    and provide a mechanism to lift a bounded-size solution from $G'$ to $G$
    while preserving its size.
    \Cref{claim:dichotomy:preserve} ensures that
    we preserve the $(k,h)$-dichotomy property in the new protrusion decomposition of $G'$, so we can repeat this process to replace each $P_i$.
    Application of \Cref{lem:prelim:replace} takes time linear in the size of the protrusion being replaced, so the total running time is linear.

    We now justify the final inequality for lifting solutions. 
    Recall that when $S$ is a feasible \fcal-deletion set  in $G$ then $\val^{\fdelsmall}_{(G,k)}(S) =\min \{|S|,k+1\}$ and $\opt_{\fdelsmall}((G,k))$ is the minimum over $\val^{\fdelsmall}_{(G,k)}(S)$.
    In these terms, we have $\opt_{\fdelsmall}((G,k)) = \opt_{\fdelsmall}((G',k))$. 
    Suppose first that $|S'| > k$ so $\val^{\fdelsmall}_{(G',k)}(S') = k+1$. 
    Then we simply return $S = V(G)$ for which $\val^{\fdelsmall}_{(G,k)}(S) = k+1$ and the inequality holds.
    Otherwise, $|S'| \le k$. We use the lifting algorithm described in \Cref{claim:dichotomy:backward} to compute a solution $S$ in $G$ so that $|S| \le |S'|$ and we have $\frac{|S|}{\opt_{\fdelsmall}((G,k))} \le \frac{|S'|}{\opt_{\fdelsmall}((G',k))}$.
\end{proof}

Note that Lemma~\ref{lem:dichotomy:compress} implies Lemma~\ref{thm:dichotomy-intro}.

\thmSparse*
\begin{proof}
    We begin with computing a constant-factor approximation to \fdel\  with \Cref{lem:prelim:apx}.
    When $c$ is the approximation factor, then the found approximate solution $X \sub V(G)$ must have size at most $c\cdot k$, otherwise we can already answer that there is no solution of size at most $k$.
    Since $\fcal$ contains a planar graph, we know that there is some constant $\beta(\fcal)$
    such that $\tw(G-X) \le \beta(\fcal)$~\cite[Theorem~$1.1$]{DBLP:conf/soda/GuptaLLM019}.
    We apply \Cref{lem:prelim:topological} to compute an $(\Oh_{\fcal,H}(k), \Oh_{\fcal,H}(1))$-protrusion decomposition of $G$.
    Next, we use \Cref{lem:dichotomy:construct} to turn it into a nice $(\Oh_{\fcal,H}(k), \Oh_{\fcal,H}(1))$-protrusion decomposition with $(k,h)$-dichotomy property, where $h$ is the maximum size of a graph in $\fcal$. 
    \Cref{lem:dichotomy:compress} allows us to compute a graph $G'$
    of size $\Oh_{\fcal,H}(k)$
      so that $\opt_\fcal(G') \le k$ if and only if $\opt_\fcal(G) \le k$. 
      Finally, observe that applying the replacement from \Cref{lem:prelim:replace} preserves being $H$-topological-minor-free so \Cref{lem:dichotomy:compress} must output a graph from this class.      
\end{proof}

\subsection{Proof of Lemma~\ref{lem:dichotomy:construct}}\label{sec:dichotomy-prop}

We first summon the known packing-covering duality for bounded-treewidth graphs.

\begin{proposition}[{\cite[Lemma 3.10]{DBLP:journals/dam/RaymondT17}}]
\label{lem:dichotomy:erdos}
    Let $\beta,\ell \in \nn$ and $G,F$ be graphs, such that $\tw(G) \le \beta$ and $F$ is connected.
    Then either $G$ contains a packing of $\ell$ disjoint minor models of $F$ or there is a set $S \sub V(G)$
     of size at most $(\beta+1)\cdot\ell$ such that $G-S$ is $F$-minor-free.

    Furthermore, for fixed $F$ there is a linear algorithm that outputs one of these two objects, when given a corresponding tree decomposition of $G$.
\end{proposition}

The following lemma utilizes the LCA closure in a tree decomposition to cover a given set by the root bag of a protrusion decomposition.

\begin{restatable}[$\bigstar$]{lemma}{lemDichLCA}
\label{lem:dichotomy:lca}
     Let $\beta \ge 1$, $G$ be a graph of treewidth at most $\beta$,
     and $S \sub V(G)$.
     Then there exists a nice $(2(\beta+1)\cdot |S|,\, 2(\beta+1))$-protrusion decomposition $(P_0,P_1,\dots,P_p)$ of $G$ with $S \sub P_0$.
     Furthermore, $(P_0,P_1,\dots,P_p)$ can be computed in linear time when a corresponding tree decomposition of $G$ is given.
\end{restatable}

Next, we will need a mechanism to refine a protrusion decomposition by another way, while preserving an $(\fcal, c)$-minor packing.

\begin{restatable}[$\bigstar$]{lemma}{lemDichProduct}
\label{lem:dichotomy:product-packing}
    Let $\fcal$ be a family of connected graphs, $c\in \nn$, $(P_0,P_1,\dots,P_p)$ be a nice $(\alpha_1, \beta_1)$-protrusion decomposition of a graph $G$, and $(Q_0,Q_1,\dots,Q_q)$ be a nice $(\alpha_2, \beta_2)$-protrusion decomposition of $G-P_0$ with an $(\fcal, c)$-minor packing.
    Then  $G$ admits a nice $(\alpha_1 + 2\alpha_2\cdot \beta_2, \,\beta_1 + \beta_2)$-protrusion decomposition with an $(\fcal, c)$-minor packing, whose root bag is $P_0 \cup Q_0$.
    Furthermore, the new decomposition and its minor packing can be computed in linear time. 
\end{restatable}

The following lemma shows that if we have a sufficiently large packing of $F$-minors in a bounded-treewidth graph, for every $F \in \fcal \sub {\cal H}^{\mathsf{conn}}_h$, then we can choose a large subset from each packing, so that together all the picked minor models are disjoint.

\begin{lemma}\label{lem:dichotomy:multi-packing}
    Let $h \in \nn$, $\fcal \sub {\cal H}^{\mathsf{conn}}_h$, $G$ be a graph such that $\tw(G) \le \beta$ and for each $F \in \fcal$ there is a disjoint packing of $6(\beta+1)\cdot|{\cal H}^{\mathsf{conn}}_h|\cdot (k+h)$ minor models of $F$ in $G$.
    Then there exists a nice $(\Oh_{h,\beta}(k),\, 2\beta+2)$-protrusion decomposition $(Q_0,Q_1,\dots,Q_\ell)$ of $G$ with an $(\fcal, k+h)$-minor packing.
    Furthermore, $(Q_0,Q_1,\dots,Q_\ell)$ and the corresponding minor packing can be constructed in time $\Oh_{h,\beta}(|V(G)|)$.
\end{lemma}
\begin{proof}
    Consider a binary tree decomposition $(T,\chi)$ of $G$ of width $\le \beta$, rooted at some node $r$. 
    It can be constructed in linear time~\cite{DBLP:journals/siamcomp/Bodlaender96}.
    For a subset $C \sub V(T)$ let $\bar\chi(C) = \chi(C) \sm \chi(V(T) \sm C)$ be the set of vertices appearing only in the bags from $C$.
    We will describe a procedure that scans $T$ bottom-up and marks two disjoint sets $M,L \sub V(T)$, both initialized as empty.
    Intuitively, a node will be marked (and added to $M$) when it is a deepest node below which we can still find a minor model that can be added to the minor packing, and $L$ will form the LCA closure of $M$.
    When we add a new node to $M$ or $L$, we will create new connected components $C$ of $T - (M \cup L)$ not containing $r$.
    To some of these components we will assign $F \in \fcal$ such that $F \le_m G[\bar\chi(C)]$.
    We will refer to the number of components assigned to $F$ as the {\em score} of $F$.
    We say that $F$ is {\em completed} when its score reaches $k+h$.

    After the $i$-th step we shall maintain the following invariants.
    \begin{enumerate}
        \item The set $M_i \cup L_i$ comprises $i$ nodes and $L_i$ is contained in $\lca (M_i)$.  \label{item:dich:count}
        \item The sum of all scores equals $|M_i|$. \label{item:dich:scores}
        \item Let $C_i^r$ denote the connected component of $T - (M_i \cup L_i)$ that contains $r$. For every connected component $C$ of $T - (M_i \cup L_i)$ different from $C_i^r$, it holds that $|N(C)| \le 2$.\label{item:dich:nc2}      
        \item For each $F \in \fcal$ either $F$ is already completed or $G[\bar\chi(C^r_i)]$ contains a packing of $3(\beta+1)\cdot\big(2|{\cal H}^{\mathsf{conn}}_h|\cdot (k+h) - i \big)$ minor models of $F$.
        \label{item:dich:remain}
    \end{enumerate}

    Clearly, the invariants hold for $i = 0$, before the first iteration.
    We will never increase a score of $F$ that has been already completed so the algorithm terminates when $|M_i| = |\fcal|\cdot (k+h)$.
    As $|M_i \cup L_i| \le |\lca(M_i)| \le 2|M_i|$ (\Cref{lem:prelim:lca}), it follows that the algorithm will perform at most $2|\fcal|\cdot (k+h) \le 2|{\cal H}^{\mathsf{conn}}_h|\cdot (k+h)$ iterations.

    Consider the $i$-th iteration, $1 \le i < 2|\fcal|\cdot(k+h)$, and suppose there are still some not completed graphs in $\fcal$. 
    Consider $t \in V(T)$ that is not equal or descendant of any $t' \in M_i \cup L_i$. 
    For such $t$ let $C^t_i \sub V(T)$ denote the set of those nodes which are descendants of $t$ but not equal to or a descendant of any $t' \in M_i \cup L_i$.    
    We look for the deepest node $t \in V(T)$ for which 
    $G[\bar\chi(C^t_i)]$
    contains a minor model of some $F \in \fcal$ that is not yet completed.
    By invariant~(\ref{item:dich:remain}), such a node $t$ must exist because for each not completed $F$ there is a packing of at least $3(\beta + 1)$ models of $F$ in $\bar\chi(C^r_i)$ whereas $\chi(r)$ can intersect at most $\beta+1$ of them, so $t=r$ satisfies the criteria.
    
    We consider two scenarios. First suppose that there are some $t_1, t_2 \in M_i$ for which $\hat t =\lcauv(t_1,t_2)$ is  descendant of $t$ and  $\hat t$ does not belong to $L_i \cup M_i$.
    Choose such  $\hat t$ that is deepest and update $L_{i+1} = L_i \cup \{\hat t\}$, $M_{i+1} = M_i$.
    In the second scenario, when there is no such  $\hat t$, we set  $M_{i+1} = M_i \cup \{ t\}$, $L_{i+1} = L_i$.
    By the choice of $t$, in the latter scenario we have created a new component $C$ of  $T - (M_i \cup L_i)$, not containing $r$, for which $F \le_m G[\bar\chi(C)]$.
    Note that $C$ can be chosen as a single component of  $T - (M_i \cup L_i)$ because $F$ is connected.
    Then we assign $F$ to~$C$.
    See \Cref{fig:dichotomy:multi-packing} for an example.

     \begin{figure}[h]
    \centering
    \includegraphics[width=\linewidth]{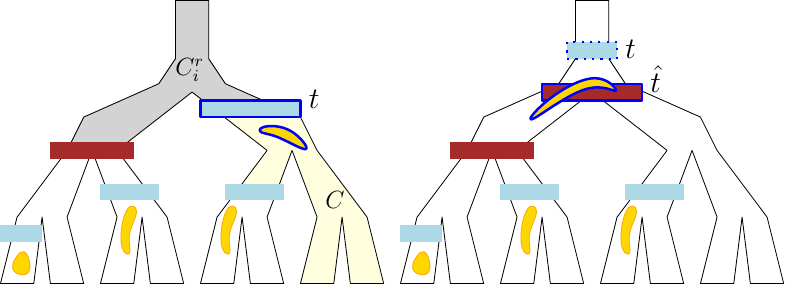}
    \caption{Illustration to \Cref{lem:dichotomy:multi-packing}.
    The bags corresponding to nodes from $M_i$ are drawn in blue and those from $L_i$ are drawn in brown.
    The already collected minor models, assigned to some components, are yellow.
    Left: A bag of $t$ selected in the $i$-th iteration and a minor model of some $F \in \fcal$ below $t$ have blue borders.
    The first scenario does not trigger so $t$ is added to $M_i$.
    This creates a component $C$ (light yellow) of $T - (M_i \cup L_i)$ to which we assign $F$.
    The component $C^r_i$ has gray background.
    Right: In this situation there is $\hat t$ below $t$ that is $\lcauv$ of two nodes from $M_{i-1}$.
    We do not add $t$ to $M_i$ but instead $\hat t$ is inserted to $L_i$.
    We do not increase any score in this scenario.
    }
    \label{fig:dichotomy:multi-packing}
\end{figure}
    
    We need to justify that all the invariants are preserved. The first two are clear.
    
    To show invariant (\ref{item:dich:nc2}), we inspect the two scenarios. In first one, we have chosen the   $\hat t =\lcauv(t_1,t_2)$ for some $t_1,t_2 \in M_i$ that does not yet belong to $L_i \cup M_i$.
    Suppose that one of the created components of $T-(M_i\cup L_i)$ below $\hat t$ has more than two neighbors.
    Then it has at least two neighbors different from $\hat t$.
    Let $s_1,s_2 \in M_i \cup L_i$ denote these neighbors and let $\hat s =\lcauv(s_1,s_2)$.
    Note that $\hat s$ must be a descendant of $\hat t$.
    By invariant (\ref{item:dich:count}), there are $s'_1,s'_2 \in M_i$ such that $\hat s =\lcauv(s'_1,s'_2)$ as well.
    This contradicts the choice of  $\hat t$ because  $\hat s$ satisfies the same condition and is located deeper.
    The same argument works in the second scenario. If some component of  $T-(M_i\cup L_i)$ created below $t$ has more than two neighbors, then $t$ must have a descendant $\hat s$ with the same property, hence it is the first scenario that must have occurred.

    To advocate invariant (\ref{item:dich:remain}) we show that for every not-yet-completed $F \in\fcal$  the 
    number of disjoint minor models of $F$ in $\bar\chi(C^r_i)$ can decrease only moderately.
    Let ${\cal M}^F$ denote the initial disjoint packing of $6(\beta+1)\cdot |{\cal H}^{\mathsf{conn}}_h|\cdot (k+h)$ minor models of $F$ in $G$.
    Let ${\cal M}^F_i \sub {\cal M}^F$ denote the subfamily of those minor models which are entirely contained in $\bar\chi(C^r_i)$.
    We aim to show that $|{\cal M}^F_{i+1}| \ge |{\cal M}^F_i| - 3(\beta+1)$; this will imply invariant (\ref{item:dich:remain}).

    Let us fix $F \in\fcal$ and let $t \in V(T)$ be the node chosen in the $i$-th iteration, in any of the two scenarios.
    Recall that we work on a binary tree decomposition so $t$ has 
    either one or two children; w.l.o.g. assume there are two of them: $t_1,t_2$.
    Furthermore, let $C_1,C_2$ be the connected components of $C^r_i - \{t\}$ that contains $t_1,t_2$, respectively.
    Let $U \sub V(G)$ induce a minor model from ${\cal M}^F_i \sm {\cal M}^F_{i+1}$.
    Suppose that $M \sub \bar\chi(C_1)$.  
    Because we have not chosen $t_1$ in the $i$-th iteration, 
    $U$ must intersect $\chi(t_1)$.
    Similarly, if $U \sub \bar\chi(C_2)$ then $U$ must intersect $\chi(t_2)$.
    If neither of these cases hold, $U$ must intersect $\chi(t)$ because $G[U]$ is connected.
    Since $|\chi(t) \cup \chi(t_1) \cup \chi(t_2)| \le 3(\beta+1)$
    and the considered minor models are disjoint, we infer that $|{\cal M}^F_i \sm {\cal M}^F_{i+1}| \le 3(\beta+1)$.
    Hence invariant (\ref{item:dich:remain}) is preserved.

    Having completed all $F \in \fcal$, the process stops.
    Let $j$ denote the number of the last iteration, let $M = M_j, L = \lca(M)$.
    Invariant (\ref{item:dich:scores}) implies that $|M| = |\fcal|\cdot (k+h)$. 
    By invariant (\ref{item:dich:nc2}) we know that $L \sm L_j \sub C^r_j$, that is, we do not add any new vertices to $L$ that would intersect a component of $T -(M_j \cup L_j)$ to which we have assigned some $F \in \fcal$.
    We define $Q_0 = \bigcup_{t \in  L} \chi(t)$; this set has at most $2(\beta+1) \cdot |M| \le 2(\beta+1) \cdot|{\cal H}^{\mathsf{conn}}_h|\cdot (k+h)$
    vertices.
    By \Cref{lem:prelim:lca} the neighborhood of each connected component of $G-Q_0$ is contained in at most two bags  of $(T,\chi)$, so it contains at most $2(\beta+1)$ vertices.
    For each $F \in \fcal$ we gather the components of  $T - L$ assigned to $F$.
    For each such component $C$ we add the set $Q_C = \bar\chi(C)$ to the protrusion decomposition.
    This gives us the $(\fcal,k+h)$-minor packing.
    The remaining components of  $T - L$ can be grouped into $\le 2|L|$ sets, each with neighborhood in at most two bags. 
    This concludes the proof.
\end{proof}

We are ready to prove the main lemma of this subsection and thus complete the proof of \Cref{lem:dichotomy:compress}.

\begin{proof}[Proof of Lemma~\ref{lem:dichotomy:construct}]
    Let $\widehat k = 6(\beta+1)\cdot |{\cal H}^{\mathsf{conn}}_h| \cdot (k+h)$.
    Due to \Cref{lem:prelims:protrusion-neighborhood} we can assume that  $(P_0,P_1,\dots,P_\ell)$ is nice, i.e., for $i \in [1,\ell]$ we have $|N_G(P_i)| \le \beta$.
    Initialize $\hcal^1 = {\cal H}^{\mathsf{conn}}_h$, $\alpha_1 = \alpha$, $\beta_1 = \beta$, and ${\cal P}^1 = (P_0,P_1,\dots,P_\ell)$, $P^1 = P_0$. 
    In the $i$-th iteration we will remove one set from $\hcal^i$ and construct a new protrusion decomposition ${\cal P}^{i+1}$ with a root bag denoted as $P^{i+1}$.
    We will maintain an invariant that $G-P^i$ is $F$-minor-free for each $F \in {\cal H}^{\mathsf{conn}}_h \sm \hcal^i$.

    In the $i$-th iteration, $i\ge 1$, we apply \Cref{lem:dichotomy:erdos} for every $F \in \hcal^i$, the graph $G-P^i_0$ and $\ell = \widehat k$.
    Suppose that for some $F \in \hcal^i$ this call returns a vertex set $S_F$ of size $\le (\beta+1) \cdot \widehat k$ for which $G-(P^i \cup S_F)$ is $F$-minor-free.
    Then we apply \Cref{lem:dichotomy:lca} to the graph $G-P^i$ and set $S_F$ to compute 
    a nice $(2(\beta+1)\cdot \widehat k, \,2(\beta+1))$-protrusion decomposition of $G-P^i$ whose root bag contains $S_F$.
    Next, we apply \Cref{lem:dichotomy:product-packing}, disregarding the parameter $\fcal$, to transform it into a nice $(\alpha_{i+1}, \beta_{i+1})$-protrusion decomposition ${\cal P}^{i+1}$ of $G$ whose root bag $P^{i+1}$ contains $S_F$ and $P^{i}$.
    We can estimate $\alpha_{i+1} = \alpha_{i} + 2(\beta+1)^2\cdot \widehat k$ and $\beta_{i+1} = \beta_i + 2(\beta + 1)$.
    We set $\hcal^{i+1} = \hcal^i \sm \{F\}$ and the invariant is preserved.

    Let $j$ denote the iteration number at which this process stops.
    Let $\fcal = \hcal^j$.
    We have $j \le  |{\cal H}^{\mathsf{conn}}_h|+1$ so $\alpha_j = \alpha + \Oh_{h,\beta}(k)$ and $\beta_j = \Oh_{h,\beta}(1)$.
    If $\fcal = \emptyset$ then we are done so let us assume otherwise.

    Since the process has stopped, \Cref{lem:dichotomy:erdos} guarantees that every $F \in \fcal$ has a packing of $\widehat k$ disjoint minor models in  $G-P^j$.
    Hence the requirements of \Cref{lem:dichotomy:multi-packing} are satisfied with respect to the graph $G-P^j$.
    Applying the lemma gives us a nice $(\Oh_{h,\beta}(k),\, 2\beta+2)$-protrusion decomposition $\cal Q$ of $G-P^j$ with an $(\fcal,k+h)$-minor packing.
    Finally, we apply \Cref{lem:dichotomy:product-packing} with respect to ${\cal P}_j$, $\cal Q$, and the constructed $(\fcal,k+h)$-minor packing for $\cal Q$.
    As a result, we obtain a nice $(\alpha + \Oh_{h,\beta}(k),\, \Oh_{h,\beta}(1))$-protrusion decomposition $\cal R$ of $G$ with an $(\fcal,k+h)$-minor packing.
    Furthermore, the root bag of $\cal R$ contains $P^j$.
    Therefore $G-P^j$ is $F$-minor-free for each $F \in {\cal H}^{\mathsf{conn}}_h \sm \fcal$ and so the above minor packing yields $(k,h)$-dichotomy property.
\end{proof}

\section{Uniform $2$-lossy polynomial compression algorithm} 
\label{sec:lossy-kernel}
In this section we prove Theorem~\ref{thm:lossy-kernel}.
Several reduction rules from the section will also come in useful in
Section~\ref{sec:lossy-protocol}.
Throughout Sections~\ref{sec:lossy-kernel} and~\ref{sec:lossy-protocol}, 
we assume that the family $\mathcal{F}$ contains a planar graph.
Fix an arbitrary $F_{\planar} \in \mathcal{F}$ which is a planar graph. The Grid Minor Theorem implies that
for any \fcal-deletion set $S$ of $G$, the treewidth of $G-S$ is  bounded in terms of $|V(F_{\planar})|$~\cite{DBLP:journals/jct/RobertsonS86}. 
We shall denote by $\eta = \eta(\fcal)$ the upper bound on the treewidth of an $\fcal$-minor-free graph.

Following the convention from \cite{DBLP:conf/stoc/LokshtanovPRS17}, we consider solution cost capped at $k+1$, i.e., we measure the cost of  an $\fcal$-deletion set $S$ in $G$ as  
$\val_{(G,k)}(S): = \min\{|S|,k+1\}$. 
This ensures that the parameter $k$ captures the difficulty of instance $(G,k)$.

Our kernelization algorithm (and protocol) starts by constructing a {\em near-protrusion decomposition} of $V(G)$ as done in~\cite[Lemma~$25$]{DBLP:conf/focs/FominLMS12}. The construction is~\cite{DBLP:conf/focs/FominLMS12} is randomized because it relies on a randomized constant-factor approximation algorithm for \fdel. This step, and hence the entire construction of~\cite[Lemma~$25$]{DBLP:conf/focs/FominLMS12}, can be made deterministic by using the deterministic constant-factor approximation for \fdel\ from \Cref{lem:prelim:apx}.

\begin{lemma}[{\cite[Lemma~$25$]{DBLP:conf/focs/FominLMS12} together with \Cref{lem:prelim:apx}}]\label{lem:near-protrusion}
    There is a polynomial-time algorithm that given an instance $(G,k)$ of \fdel, when $\mathcal{F}$ contains a planar graph, either reports correctly that $(G,k)$ is a no-instance of \fdel, 
    or computes $X,Z \subseteq V(G)$, $X \cap Z = \emptyset$ such that:
    \begin{itemize}
        \item $|X|=\Oh_{\eta}(k)$, $|Z| =\Oh_{\eta}(k^3)$, 
     $X$ is an \fcal-deletion set of $G$,
        \item for each connected component $C$ of $G-(X\cup Z)$, $|N_G(C) \cap Z| \leq 2(\eta+1)$, and
        \item for each connected component $C$ of $G-(X\cup Z)$, and distinct $u,v \in N_G(C) \cap X$, 
        there are at least $k+\eta+2$ vertex-disjoint paths from $u$ to $v$ in $G$. 
    \end{itemize}
\end{lemma}

Throughout the rest of this section, $X,Z$ stand for the sets returned by Lemma~\ref{lem:near-protrusion} for input $(G,k)$.
We now design a $(1+\epsilon)$-lossy strict reduction rule
(see \Cref{sec:prelims-lossy})
that bounds the size of the neighborhood of each connected component of $G-(X\cup Z)$.

\begin{lossyredrule}\label{lrr:protrusion}
    Let $\epsilon >0$ and 
    let $C$ be a connected component of $G-(X \cup Z)$ such that $|N_G(C) \cap X| \geq \frac{(1+\epsilon)(\eta+1)}{\epsilon} $.
    In this case the reduction algorithm takes as input an instance $(G,k)$ of \fdel\ and outputs the instance $(G':= G-(N_G(C)\cap X),\,k':= k-|N_G(C) \cap X|)$ of \fdel.

    The solution lifting algorithm takes as input instances $(G,k), (G',k')$ and a set $S' \subseteq V(G')$ and outputs the set $S \subseteq V(G)$, where $S:=S'\cup (N_G(C) \cap X)$.
\end{lossyredrule}

\begin{lemma}[$\bigstar$]\label{lem:lrr-one-correct}
    \Cref{lrr:protrusion} is a $(1+\epsilon)$-lossy strict reduction rule.
\end{lemma}

After an application of \Cref{lrr:protrusion}, some vertices from $X$ may get deleted. But the remaining set $X$, together with $Z$, still satisfies the properties of Lemma~\ref{lem:near-protrusion} in the resulting graph.
We keep the variable $X$ to denote this set in the new graph.

\begin{lemma}[$\bigstar$]\label{lem:protrusion-decomposition}
    When \Cref{lrr:protrusion} is no longer applicable, then there is an $(\alpha,\Oh(\eta / \epsilon))$-protrusion decomposition
     $(P_0,P_1, \ldots, P_{\ell})$ of $G$ for some $\alpha$, where $P_0:= X \cup Z$ and for each $i \in [1,\ell]$, $P_i$ is a connected component of $G-(X \cup Z)$.
     Furthermore, $(P_0,P_1, \ldots, P_{\ell})$ can be computed in polynomial time.
\end{lemma}

Let $(P_0,\ldots, P_{\ell})$ be the protrusion decomposition obtained from Lemma~\ref{lem:protrusion-decomposition}. 
We now define the {\em augmented graph $G_{\flow}$ of $G$} which is a supergraph of $G$ on the same vertex set as $G$, and which is obtained by additionally
adding the edge set $\{uv: u,v \in P_0, \text{ there exist at least } k+\eta+2 \allowbreak \text{ internally vertex-disjoint \allowbreak  paths \allowbreak  from } \allowbreak  u  \allowbreak \text{ to } \allowbreak  v \allowbreak  \text{ in } \allowbreak G\}$ in $G$.

\begin{lemma}[$\bigstar$]\label{lem:fdel-tweta-flow}
    The following two (in)equalities hold:
    \begin{enumerate}
        \item $\opt_{\twetadelsmall}((G,k))  =  \allowbreak  \opt_{\twetadelsmall}((G_{\flow},k))$, 
       \item $\opt_{\fdelsmall}((G,k)) \geq \opt_{\twetadelsmall}((G_{\flow},k))$.
       \end{enumerate}
\end{lemma}

We partition  $\{P_1, \ldots, P_{\ell}\}$ into two groups based on the neighborhood of $P_i$ in ${G_{\flow}}$.
Let $\mathcal{P}_{\clique}=\{P_i: G_{\flow}[N(P_i)] \text{ is a clique or $N(P_i)$ is empty}\}$, which we call  the set of {\em simplicial parts} of the protrusion decomposition.
Let $\mathcal{P}_{\nonclique}$
be the set of remaining parts, referred to as {\em non-simplicial parts}. 
Note that we cannot simply discard $P_i$ for which $N(P_i)=\emptyset$ in the case of \fdel\ when $\fcal$ contains disconnected graphs.

Observe that when $u,v\in P_0$ and $uv \not\in E(G_{\flow})$ then $\{u,v\}$ can belong to at most $k+\eta+1$ sets $N(P_i)$.
This argument leads to the following estimation.

\begin{lemma}[$\bigstar$]\label{lem:non-clique}
    $|\mathcal{P}_{\nonclique}| = \Oh_{\eta}(k^5)$.
\end{lemma}

We show that ignoring the simplicial parts yields a 2-lossy compression reduction rule.

\begin{lossyredrule}\label{lrr:two}
    Suppose we are given an instance $(G,k)$ of \fdel\ with the corresponding protrusion decomposition $(P_0,P_1, \ldots, P_{\ell})$. 
    The reduction algorithm outputs $(G':= G_{\flow}[V(G) \setminus \bigcup_{P_i \in \mathcal{P}_{\clique}} P_i],\, k':=k)$.

    The solution lifting algorithm takes a solution $S'$ of the \twetadel\ problem on the instance $(G',k')$. 
    It computes an optimum-sized \fcal-deletion set $S''$ of $G-S'$ in $\Oh_{\fcal, \eta}(n^{\Oh(1)})$ time using the algorithm of \Cref{lem:f-del-on-tw} 
    (we argue in Lemma~\ref{lem:lrr-two-correct} that this is possible). 
    It then outputs $S:= S' \cup S''$ as a solution on \fdel\ for the instance $(G,k)$.
\end{lossyredrule}

Specifically, we prove that the simplicial parts can be added to $G'-S'$ while keeping the treewidth bounded by $\Oh(\eta)$.
Hence $S''$ can be computed in polynomial time.

\addtocounter{theorem}{1}
\begin{lemma}[$\bigstar$]\label{lem:lrr-two-correct}
    \Cref{lrr:two} is a $2$-lossy compression reduction rule.
\end{lemma}

After the exhaustive application of \Cref{lrr:protrusion} and~\ref{lrr:two}, it only remains to bound the size of each part in $\mathcal{P}_{\nonclique}$, which can be done using Lemma~\ref{lem:dichotomy:compress}. 
The formal proof of Theorem~\ref{thm:lossy-kernel} is given next.

\lossykernel*

\begin{proof}

    We describe a polynomial time $2$-lossy compression algorithm below. Refer to Figure~\ref{fig:lossy-kernel} in the appendix. 
    Given an instance $(G,k)$ of \fdel,
    first compute a near-protrusion decomposition on the instance $(G,k)$ in polynomial time using Lemma~\ref{lem:near-protrusion}. This gives sets $X,Z$ with the properties listed in Lemma~\ref{lem:near-protrusion}.
    Set $\epsilon =1$ and
    for each connected component $C$ of $G-(X \cup Z)$, it checks whether \Cref{lrr:protrusion} is applicable. 
    Let $(G^1,k_1)$ be the instance returned by the reduction algorithm of \Cref{lrr:protrusion} at the end of its exhaustive application. 
    From Lemma~\ref{lem:protrusion-decomposition}, $G^1$ has an
    $(\alpha, \Oh(\eta))$-protrusion decomposition $(P_0,\ldots, P_{\ell})$
     for some (unbounded) $\alpha$
     and 
     $|P_0| = \Oh_{\eta}(k^3)$.
    
    Let $\mathcal{P}_{\nonclique}$ be the set of non-simplicial parts of $(P_1, \ldots, P_{\ell})$.
    From Lemma~\ref{lem:non-clique}, $|\mathcal{P}_{\nonclique}|=\Oh_{\eta}(k^5)$.
    Run the reduction algorithm of \Cref{lrr:two} on the instance $({G^1}_{\flow}[V(G^1) \setminus \bigcup_{P_i \in \mathcal{P}_{\clique}}P_i],k_1)$ of \fdel. 
    It returns an instance $(G^2,k_1)$ of \twetadel\ such that $G^2$ has an $(\Oh_{\eta}(k^5), \Oh(\eta))$-protrusion decomposition.
    Let $(G^3,k_1)$ be the instance returned by the algorithm of Lemma~\ref{lem:dichotomy:compress} on input $(G^2,k_1)$. Observe that $k_1 \leq k$.
    From Lemma~\ref{lem:dichotomy:compress}
    the number of vertices of $G^3$ is bounded by
    $\Oh_{\mathcal{F}}(k^5)$. Observe that if $G^3$ indeed has an at most $k$-sized vertex set whose deletion results in a graph of treewidth at most $\eta$, then treewidth of $G^3$ has to be at most $k+\eta$. Thus, the number of edges of $G^3$, and hence the size of $G^3$, is at most $(k+\eta)$ times the number of vertices of $G^3$. Thus, the size of $G^3$ is $\Oh_{\mathcal{F}}(k^6)$, otherwise $G^3$ has no $k$-sized solution.

    For any $\beta \geq 1$, 
    run the $\beta$-compression oracle for \twetadel\ on the instance $(G^3,k_1)$ and let 
     $S_3$ be the returned $\beta$-approximate solution of the instance $(G^3,k_1)$ of \twetadel.
    The algorithm of Lemma~\ref{lem:dichotomy:compress} provides a lifting algorithm that, in polynomial time, outputs a $\beta$-approximate solution $S_2$ for the instance $(G^2,k_1)$ of \twetadel.
    Given the $\beta$-approximate solution $S_2$ for the instance $(G^2,k_1)$ of \twetadel, from Lemma~\ref{lem:lrr-two-correct},
    the solution lifting algorithm of \Cref{lrr:two},  
    outputs a $2 \cdot \beta$-approximate solution $S_1$ for the instance $(G^1,k_1)$ of the \fdel\ problem.
    Finally since \Cref{lrr:protrusion} is a $(1+\epsilon)$-lossy strict reduction rule and $\epsilon=1$, using the solution lifting algorithm of \Cref{lrr:protrusion} and $S_1$, one can obtain a $2$-approximate solution for the instance $(G,k)$ of \fdel.

    Clearly, in the special case of \twetadel\ the above constitutes a $2$-lossy kernelization algorithm.
\end{proof}

\section{Uniform $(1+\epsilon)$-lossy polynomial compression protocol} 
\label{sec:lossy-protocol}
This section is devoted to the proof of Theorem~\ref{thm:lossy-protocol}.
Assume that we are given an instance $(G,k)$ of \fdel, $\epsilon >0$, 
and an $(\alpha, \Oh(\eta / \epsilon))$-protrusion decomposition $(P_0,P_1, \ldots, P_{\ell})$ of $G$ from Lemma~\ref{lem:protrusion-decomposition}, for some unbounded $\alpha$, with $|P_0|=\Oh_{\eta}(k^3)$.
Recall that the parts $\{P_1, \ldots, P_{\ell}\}$ of this protrusion decomposition were partitioned into two sets: $\mathcal{P}_{\clique}$ (simplicial parts) and $\mathcal{P}_{\nonclique}$ (non-simplicial parts).
Moreover, Lemma~\ref{lem:non-clique} estimates  $|\mathcal{P}_{\nonclique}|$ as $\Oh_{\eta}(k^5)$.

\begin{figure}[ht]
    \begin{center}
    \includegraphics[scale=0.33]{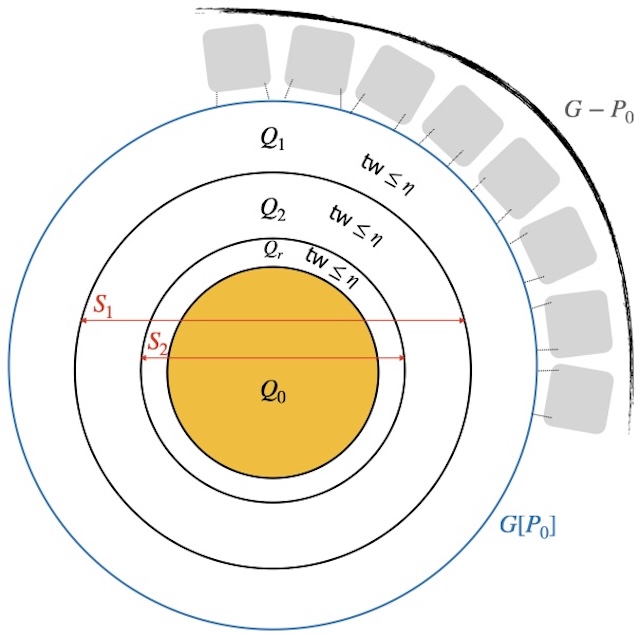}
    \caption{Illustration to \Cref{lem:partition-mod}. The protocol starts by finding a $\beta$-approximate solution $S_1$ to the \twetadel\ problem in $G[P_0]$. Then $Q_1$ is the set obtained after deleting $S_1$ from $P_0$; note that $\tw(G[Q_0]) \le \eta$. It then repeats the process of finding a solution to the \twetadel\ problem inside the old solution $S_1$ and continues doing so until the obtained new solution is a large fraction of the previous solution. Afterwards, \Cref{lrr:three} removes the final solution from the graph and compresses the remaining graph using \Cref{lem:bound-simplicial}.}
    \label{fig:lossy-protocol-oracle}
    \end{center}
\end{figure}

\begin{lemma}\label{lem:partition-mod}
    Consider $\epsilon > 0$ and $\beta \in \nn$.
    Let $r = 1+ \left \lceil \log_{1+\epsilon} \left ( \frac{1}{\epsilon}\right )  \right \rceil$.
    There exists a polynomial-time protocol which has access to a $\beta$-approximate kernelization oracle $\mathsf{O}$ for \twetadel\ of capacity $|P_0|$, 
    and when given $(G_{\flow}[P_0],k)$ as input,
    it makes at most $r$ calls 
    to the oracle $\mathsf{O}$,
    and outputs a partition of $P_0=(Q_0, \ldots, Q_r)$,
     such that
    \begin{enumerate}
        \item for each $i \in [1,r]$, treewidth of ${G}_{\flow}[Q_i]$ is at most $\eta$ ($Q_i$ is allowed to be empty), and
        \item \begin{enumerate}
            \item 
        $|Q_0| \leq (1+\epsilon)\cdot \beta \cdot   \opt_{\twetadelsmall}((G_{\flow}[Q_0],k))$ 
        or \label{item:halt:short}
            \item $|Q_0| \le \allowbreak \epsilon \cdot  \beta \cdot \opt_{\twetadelsmall}((G_{\flow}[P_0],k))$. \label{item:halt:long}
        \end{enumerate}
    \end{enumerate}
\end{lemma}

\begin{proof}
    We design an iterative procedure with at most $r$ steps. 
    See Figure~\ref{fig:lossy-protocol-oracle} for a visualization.
    Initialize $Q_0:=P_0$ and for each $i \in [1,r]$, set $Q_i:=\emptyset$.
    Let $S_1$ be the set returned by the oracle $\mathsf{O}$ on input $(G_{\flow}[P_0],k)$. 
    Then  $S_1$ is a $\beta$-approximate solution to the instance $(G_{\flow}[P_0],k)$ of \twetadel, and so $|S_1| \leq \beta \cdot \opt_{\twetadelsmall}((G[P_0], \allowbreak k))$.
    
    In each iteration, we maintain the following invariants: (i) $(Q_0,Q_1, \ldots, Q_r)$
    is a partition of $P_0$, (ii) for each $i \in [1, r]$, the treewidth of $G_{\flow}[Q_i]$ is at most $\eta$, and (iii) at the end of the $j$-th iteration, $|Q_0| \le \left ( \frac{1}{1+\epsilon} \right )^{j-1} |S_1|$.

    In the $j$-th iteration, $j \in [1,r]$,
    the protocol calls the oracle $\mathsf{O}$ on the instance $(G_{\flow}[Q_0],k)$ and receives a $\beta$-approximate solution $S_j$ to \twetadel\ on $(G_{\flow}[Q_0],k)$.
    Then $|S_j| \leq \beta \cdot \opt_{\twetadelsmall}((G_{\flow}[Q_0],k))$.
    The protocol proceeds with the following case distinction. 
    \begin{enumerate}
        \item If $|Q_0| \leq (1+\epsilon) \cdot |S_j|$, then the protocol terminates and outputs the partition $(Q_0, \ldots, Q_r)$ which satisfies condition (\ref{item:halt:short}).
        \item Otherwise, $|Q_0| > (1+\epsilon) \cdot |S_j|$. 
        In this case, we set $Q_j:= Q_0 \setminus S_j$.
        Clearly, the treewidth of $G_{\flow}[Q_j]$ is at most $\eta$.
        Next, we set $Q_0:= S_j$. 
        
        We now verify that the invariants are satisfied.
        Clearly $(Q_0,Q_1, \ldots, Q_r)$ is a partition of $P_0$ and the treewidth of each $G_{\flow}[Q_i]$, $i \in [1,r]$, is at most $\eta$.
        If $j=1$ then the invariant (iii) holds trivially.
        Otherwise, due to case distinction, the size of $Q_0$ drops by factor $(1+\eps)$,
        implying $|Q_0| \le \left ( \frac{1}{1+\epsilon} \right )^{j-1} |S_1|$.
    \end{enumerate}

    If the first case occurs before the $r$-th iteration, then we are done. 
   Otherwise, the protocol terminates after the $r$-th iteration. 
   The invariant (iii) implies
   $|Q_0| \le \left ( \frac{1}{1+\epsilon} \right )^{r-1} |S_1|$.
   Substituting the definitions of $r$ and $S_1$ yields condition (\ref{item:halt:long}): $|Q_0| \le \epsilon \cdot \beta \cdot \opt_{\twetadelsmall}((G_{\flow}[P_0],k))$. 
   The lemma follows.
\end{proof}

We argue that one can remove $Q_0$ from the graph while paying a small loss in accuracy.

\begin{lossyredprotocol}\label{lrr:three}
    Let $(G,k)$ be an instance of \fdel\ where \Cref{lrr:protrusion} has been exhaustively applied. Let $(P_0, \ldots, P_{\ell})$ be the protrusion decomposition of $(G,k)$ from Lemma~\ref{lem:protrusion-decomposition}.
    Let $(Q_0, Q_1, \ldots, Q_r)$ be the partition of $P_0$ obtained
    from Lemma~\ref{lem:partition-mod}, on input $(G_{\flow}[P_0],k)$.
    The reduction protocol outputs $(G':=G-Q_0, k':= k - |Q_0|)$.

    The solution lifting algorithm takes as input the instances $(G,k), (G',k')$ and $S' \subseteq V(G')$ which is a $\beta$-approximate solution for the instance $(G',k')$ of \fdel\, and outputs $S:= S' \cup Q_0$ as a solution for $(G,k)$ of \fdel.
\end{lossyredprotocol}

\begin{lemma}[$\bigstar$]\label{lem:lrr-three} 
    \Cref{lrr:three} is a $(1+\epsilon)$-lossy reduction protocol.
\end{lemma}

On the other hand, discarding $Q_0$ allows us to bound the number of simplicial parts.
\addtocounter{theorem}{1}

\begin{lemma}\label{lem:bound-simplicial}
    After the application of \Cref{lrr:three}, 
   the number of distinct cliques in $G_{\flow}[P_0]$ is $\Oh_{\eta}(k^{3r})$, where $r$ is as defined in Lemma~\ref{lem:partition-mod}.
\end{lemma}
\begin{proof}
 Let $(Q_0, \allowbreak Q_1, \allowbreak \ldots, \allowbreak Q_r)$ be the partition of $P_0$ obtained from Lemma~\ref{lem:partition-mod}.
    After the application of \Cref{lrr:three},
    we can assume that $Q_0 = \emptyset$.
    Recall that for each $i \in [1,r]$, treewidth of $G_{\flow}[Q_i]$ is at most $\eta$. From Proposition~\ref{prop:number-of-cliques}, the number of distinct cliques in $G_{\flow}[Q_i]$ is $\Oh_{\eta}(|Q_i|) \leq \Oh_{\eta}(|P_0|) \leq \Oh_{\eta}(k^3)$. 
    Therefore, the number of distinct cliques in $G_{\flow}[P_0]$ is at most $\prod_{i =1}^r \Oh_{\eta}(k^3) = \Oh_{\eta}(k^{3r})$.
\end{proof}

\lossyprotocol*

\begin{proof}
    Fix any $\epsilon >0$.
    We describe a polynomial time $(1+\epsilon)$-lossy compression protocol of capacity $\Oh_{\eta}(k^7) + \Oh_{\eta}(k^{3r})$ and $1+r$ rounds.
    For any $\beta \geq 1$
    it has access to $\beta$-kernelization oracles for \fdel\ and \twetadel\ of capacities $\Oh_{\eta}(k^6) + \Oh_{\eta}(k^{3r +1})$.

    Given an instance $(G,k)$ of \fdel,
    first compute a near-protrusion decomposition on the instance $(G,k)$ in polynomial time using Lemma~\ref{lem:near-protrusion}. This gives sets $X,Z$ with the properties listed in Lemma~\ref{lem:near-protrusion}.
    For each connected component $C$ of $G-(X \cup Z)$, it checks whether \Cref{lrr:protrusion} is applicable. 
    Let $(G^1,k_1)$ be the instance returned by the reduction algorithm of \Cref{lrr:protrusion} at the end of its exhaustive application. 
    From Lemma~\ref{lem:protrusion-decomposition}, $G^1$ has an
    $(\alpha, \Oh(\eta/ \epsilon))$-protrusion decomposition $(P_0,\ldots, P_{\ell})$
     for some (unbounded) $\alpha$
     and 
     $|P_0| = \Oh_{\eta}(k^3)$.
    Run \Cref{lrr:three} on the instance $(G^1,k_1)$ to compute $(G^2,k_2)$.

    Consider the augmented graph $G^2_{\flow}$ of $G^2$.
     The number of $P_i$, $i \in [1,\ell]$, those neighborhood in $G^2_{\flow}$ is non-empty and not a clique is bounded by $\Oh_{\eta}(k^7)$ from Lemma~\ref{lem:non-clique}.
    For the remaining $P_i$, we partition them into sets such that each part in a fixed set has the same neighborhood in $G^2_{\flow}$ and any two sets have distinct neighborhood in $G^2_{\flow}$. Let the resulting sets be $R_1, \ldots, R_{\rho}$. 
    Note that for each $i \in [1,r]$, $N[R_i]$ is an $\Oh(\eta/ \epsilon)$-protrusion in $G$ and $N_{G^2_{\flow}}(R_i)$ is a clique in $G^2_{\flow}$.
    From Lemma~\ref{lem:bound-simplicial} we infer that $\rho= \Oh_{\eta}(k^{3r})$.

    We obtain that $G^2$ has a $(\delta, \gamma)$-protrusion decomposition, where $\delta= \Oh_{\eta}(k^5) + \Oh_{\eta}(k^{3r})$ and $\gamma= \Oh(\eta/ \epsilon)$.
    Let $(G^3,k_3)$ be the instance obtained from Lemma~\ref{lem:dichotomy:compress} on input $(G^2,k_2)$. 
    Then the number of vertices of $G^3$ is $\Oh_{\mathcal{F}}(\delta)$. Also, if $G^3$ has a $k$-sized solution then the treewidth of $G^3$ is at most $k+\eta$ and therefore the number of edges of $G^3$ is at most $\Oh_{\mathcal{F}}(k \cdot \delta)$.
    Finally run the $\beta$-kernelization algorithm for the problem \fdel\ on the instance $(G^3,k_3)$. Let $S_3$ be a $\beta$-approximate solution for the problem \fdel\ on the instance $(G^3,k_3)$.

    Using the polynomial-time algorithm of Lemma~\ref{lem:dichotomy:compress}, compute a $\beta$-approximate solution $S_2$ for the instance $(G^2,k_2)$ of \fdel.
    We use the solution lifting algorithm of Lemma~\ref{lem:lrr-three} with $S_2$:
    let $S_1$ be the corresponding $(1+\epsilon)\cdot \beta$-approximation solution obtained for the instance $(G^1,k_1)$ of \fdel.
    Finally, the solution lifting algorithm of Lemma~\ref{lem:lrr-one-correct} takes $S_1$ and gives a $(1+\epsilon)\cdot \beta$-approximation solution for the instance $(G,k)$. 
\end{proof}

\section{Conclusion}\label{sec:conclude}

In this work, we have presented new techniques to turn a near-protrusion decomposition into a protrusion decomposition (by sacrificing accuracy) and to process the latter in the presence of disconnected forbidden minors.
A main follow-up question is whether one really needs multiple calls to the oracle to obtain a uniform $(1+\eps)$-approximate kernel for \twetadel\ (and other cases of \fdel).
In other words, can we improve the uniform lossy kernelization protocol to uniform lossy kernelization?
Secondly, our protocol requires the oracle to handle instances of size $\Oh_{\eta,\eps} (k^{f(\eps)})$ for some function $f$.
Can we obtain a lossy kernel/protocol that is uniform also with respect to $\eps$?

\bibliography{references}

\appendix

\section{Preliminaries Continued}\label{app:prelims}

A graph $H$ is a {\em minor} of a graph $G$ (denoted $H \le_m G$) if it can be obtained from $G$ by a (possibly empty) series of vertex deletions, edge deletions, and edge contractions.
It is well-known that if $H \le_m G$ then there exists a {\em minor model} $\phi \colon V(H) \to 2^{V(G)}$ such that (i)~for each $v \in V(H)$ the graph $G[\phi(v)]$ is non-empty and connected, (ii) $\phi(u) \cap \phi(v) =\emptyset$ for $u\ne v \in V(H)$, and (iii)
if $uv \in E(H)$ then there exists an edge between $\phi(u)$ and $\phi(v)$.
For a graph family $\fcal$, 
a set $S \subseteq V(G)$ is called an 
\emph{\fcal-deletion set}
in $G$, if $G-S$ has no graph from $\mathcal{F}$ as a minor.

\begin{proposition}[{\cite[Fact 1]{DBLP:conf/stoc/LokshtanovPRS17}}]\label{prop:fact} 
    For any $a,b,c,d \in \mathbb{R}^+$, $\min\left \{ \frac{a}{c}, \frac{b}{d}\right \}   \leq \frac{a+b}{c+d}   \leq \max \left \{ \frac{a}{c} , \frac{b}{d} \right \}$.
\end{proposition}

\begin{definition}[Tree Decomposition]
    For any graph $G$, a tree decomposition of $G$ is a pair $(T,\beta)$ where $T$ is a tree and $\beta: V(T) \to 2^{V(G)}$ satisfying the following properties.
    \begin{itemize}
        \item For each $v \in V(G)$, there exists $t \in V(T)$ such that $v \in \beta(t)$.
        \item For each $uv \in E(G)$, there exists $t \in V(T)$ such that $u,v \in \beta(t)$.
        \item For each $v \in V(G)$, $T[\{t: t \in V(T), v \in \beta(t)\}]$ is connected.
    \end{itemize}
\end{definition}

For any $t \in V(T)$, the set $\beta(t)$ is called a {\em bag} of the tree decomposition $(T,\beta)$.
The {\em  width} of $(T,\beta)$ is $\max_{t \in V(T)} |\beta(t)-1|$.
The {\em treewidth} of $G$ is the minimum width over all tree decompositions of $G$.

\begin{proposition}\label{prop:number-of-cliques}
An $n$-vertex graph of treewidth $\tw$ has $\Oh(2^{\tw} \cdot \tw \cdot n)$ distinct cliques.
\end{proposition}
\begin{proof}
    An $n$-vertex graph $G$ of treewidth $\tw$ has a tree decomposition of width $\tw$ and $\Oh(\tw \cdot n)$ nodes~\cite[Lemma~$7.4$]{DBLP:books/sp/CyganFKLMPPS15}.
    Since each clique of $G$ is contained inside some bag of this tree decomposition~\cite[Lemma~$12.3.5$]{DBLP:books/daglib/0030488},
    the number of distinct cliques are at most $2^{\tw+1}$ times the number of nodes in a tree decomposition.
\end{proof}

\begin{proposition}[{\cite[Theorem~$4$]{DBLP:journals/siamdm/BasteST20}}]\label{lem:f-del-on-tw}
    When $\mathcal{F}$ is a finite family of graphs and contains at least one planar graph,
    then \fdel\ can be solved in time $2^{\Oh_{\mathcal{F}} (\tw \log \tw)} \cdot n$ on an $n$-vertex graph of treewidth $\tw$.
\end{proposition}

\begin{proposition}[{\cite[Corollary~$1.1$]{DBLP:conf/soda/GuptaLLM019}}]
\label{lem:prelim:apx}
When $\mathcal{F}$ is a finite family of graphs that contains at least one planar graph,
    then \fdel\ admits a polynomial-time 
    $\Oh_{\mathcal{F}}(1)$-approximation algorithm.
\end{proposition}

\begin{proposition}[\cite{DBLP:journals/siamdm/GarneroPST15, kim15linear}]
\label{lem:prelim:topological}
Let $\beta \in \nn$, $H$ be an $h$-vertex
graph, and $G$ be an $n$-vertex $H$-topological-minor-free graph.
There is a linear-time algorithm, that given $G$ and $X \sub V(G)$ satisfying $\tw(G-X) \le \beta$, outputs an $(\alpha_H \cdot \beta \cdot |X|, \,2\beta + h)$-protrusion decomposition of $G$, where $\alpha_H$ depends only on $H$. 
\end{proposition}

\subsection{Lossy Kernelization}\label{sec:prelims-lossy}

Let $\Pi$ be a parameterized minimization vertex-deletion graph problem and $I=(G,k)$ be an instance of $\Pi$ where $G$ is a graph and $k$ is the parameter and the solution budget. 
In our settings, $\Pi$ will be a placeholder for one of these two problems: \fdel\ (\fdelsmall) and \twetadel (\twetadelsmall).
In the \twetadel\ problem, $\eta$ is a fixed positive integer constant, 
the input is an undirected graph $G$ and a positive integer $k$, 
and the goal is to find at most $k$ vertices $X \subseteq V(G)$, such that the treewidth of $G-X$ is at most $\eta$. 

As the parameter $k$ is intended to capture the difficulty of the approximation task, we consider the solution costs capped at $k+1$, following the convention from \cite{DBLP:conf/stoc/LokshtanovPRS17}.

Let $I=(G,k)$ be an instance of $\Pi$.
For any $S \subseteq V(G)$, define   
$\val^{\Pi}_I(S): = \min\{|S|,k+1\}$ if $S$ is a solution of $\Pi$ in $G$,
otherwise define $\val^{\Pi}_I(S): = +\infty$.
Finally, $\opt_{\Pi}(I):= \min_{S \subseteq V(G)} \{\val^{\Pi}_{I}(S)\}$.
A set $S \subseteq V(G)$ is called a {\em $(\val^{\Pi}_I(S) / \opt_{\Pi}(I))$-approximate solution} of $\Pi$ for the instance $I$.

\begin{definition}[$\alpha$-lossy strict reduction rule \cite{DBLP:conf/stoc/LokshtanovPRS17}]
    Let $\Pi$ be a parameterized minimization problem. 
    Let $\alpha \geq 1$.
    An \emph{$\alpha$-lossy strict reduction rule} for $\Pi$ consists of a pair of polynomial-time algorithms, called the \emph{reduction algorithm}, and the \emph{solution lifting} algorithm.

    The reduction algorithm takes as input an instance $I=(G,k)$ of $\Pi$ and outputs another instance $I'=(G',k')$ of $\Pi$ such that $k'\leq k$.

    The solution lifting algorithm takes as input instances $I,I'$ and a set $S' \subseteq V(G')$,
    and outputs a a set $S \subseteq V(G)$ such that 

    $$\frac{\val^{\Pi}_I(S)}{\opt_{\Pi}(I)} \leq \max \left \{\alpha, \frac{\val^{\Pi}_{I'}(S')}{\opt_{\Pi}(I')} \right\}.$$ 
\end{definition}

A more general notion of an $\alpha$-lossy kernelization algorithm, called an {\em $\alpha$-lossy kernelization protocol} was defined in~\cite[Definition~$5$]{DBLP:conf/esa/FominLL0TZ23}. 
Below we define a more general notion called an {\em $\alpha$-lossy compression protocol}.

For any $\beta \geq 1$, $s \in \mathbb{N}$ and any parameterized minimization problem $\Pi$, 
a {\em $\beta$-approximate kernelization oracle for $\Pi$ of capacity $s$}, takes as input an instance $I=(G,k)$ of $\Pi$ where $|I|,k \leq s$, and outputs a $\beta$-approximate solution $S$ of the instance $I$ of $\Pi$.

\begin{definition}[$\alpha$-lossy compression protocol]
    Let $\alpha, \beta \geq 1$, $f: \mathbb{N} \times \mathbb{N} \to \mathbb{N}$ and $g: \mathbb{N} \times \mathbb{N} \to \mathbb{N}$.
    An {\em $\alpha$-lossy compression protocol} for a parameterized minimization problem $\Pi$,
    of {\em call size} $f(k,\alpha)$, and $g(k,\alpha)$ {\em rounds}, 
    is given as input an instance $(G,k)$ of $\Pi$ 
    and has access to $\beta$-approximate kernelization oracles $\mathsf{O}$ and $\mathsf{O}'$ of $\Pi$ and some other parameterized minimization problem $\Pi'$ respectively,
    of capacity $f(k,\alpha)$.
    It performs at most $g(k,\alpha)$ calls to $\mathsf{O}$ and/or $\mathsf{O}'$ 
    and any other operations in polynomial time, 
    and outputs an $(\alpha\cdot \beta)$-approximate solution $S$ for the problem $\Pi$ of the instance $(G,k)$.
\end{definition}

In our setting, $\Pi$ is either the \fdel\ problem of the \twetadel\ problem, and $\Pi'$ will be the \twetadel\ problem.
An {\em $\alpha$-lossy kernelization protocol} for $\Pi$ is an $\alpha$-lossy compression protocol when $\Pi' = \Pi$.
An {\em $\alpha$-lossy kernelization algorithm} is an $\alpha$-lossy kernelization protocol when the number of rounds is $1$.

For convenience, we also define an analogue of $\alpha$-lossy (non-strict) reduction rule to $\alpha$-lossy compression protocol. This will be useful to present cleanly the intermediate arguments.

\begin{definition}[$\alpha$-lossy reduction protocol]
    Let $\alpha, \beta \geq 1$, $f: \mathbb{N} \times \mathbb{N} \to \mathbb{N}$ and $g: \mathbb{N} \times \mathbb{N} \to \mathbb{N}$.
    An {\em $\alpha$-lossy reduction protocol} for a parameterized minimization problem $\Pi$,
    of {\em call size} $f(k,\alpha)$, and $g(k,\alpha)$ {\em rounds}, 
    is given as input an instance $(I,k)$ of $\Pi$ 
    and has access to $\beta$-approximate kernelization oracles $\mathsf{O}$ and $\mathsf{O}'$ of $\Pi$ and some other parameterized minimization problem $\Pi'$ respectively,
    of capacity $f(k,\alpha)$.
    It performs at most $g(k,\alpha)$ calls to $\mathsf{O}$ and/or $\mathsf{O}'$ 
    and any other operations in polynomial time, 
    and outputs an instance $(I',k')$ of $\Pi$.

    It is accompanied by a solution lifting algorithm which is a polynomial time algorithm and additionally has access to a $\beta$-approximate solution $S'$ of $\Pi$ of the instance $(I',k')$ and outputs a solution $S$ to the instance $(I,k)$ of $\Pi$ such that 
    $$\frac{\val^{\Pi}_I(S)}{\opt_{\Pi}(I)} \leq  \alpha \cdot \beta.$$ 
\end{definition}

\section{Missing proofs}

\subsection{Proof of Lemmas~\ref{lem:dichotomy:lca} and \ref{lem:dichotomy:product-packing}}

\iffalse
\begin{lemma}[$\bigstar$]
\label{lem:dichotomy:lca}
     Let $\beta \ge 1$, $G$ be a graph of treewidth at most $\beta$,
     and $S \sub V(G)$.
     Then there exists a nice $(2(\beta+1)\cdot |S|,\, 2(\beta+1))$-protrusion decomposition $(P_0,P_1,\dots,P_p)$ of $G$ with $S \sub P_0$.
     Furthermore, $(P_0,P_1,\dots,P_p)$ can be computed in linear time when a corresponding tree decomposition of $G$ is given.
\end{lemma}
\fi
\lemDichLCA*
\begin{proof}
    Let $(T,\chi)$ be a rooted tree decomposition of $G$ of width $\le \beta$.
    We construct set $S_T \sub V(T)$ by choosing for each $v \in S$ some node $t \in V(T)$ for which $v \in \chi(t)$.
    Let $L = \overline{\mathsf{LCA}}(S_T)$.
    By \Cref{lem:prelim:lca} we know that $|L| \le 2|S_T| \le 2|S|$ and each connected component $C$ of $T-L$ satisfies $|N_T(C)| \le 2$.
    Let $\cal C$ denote the family of those components of $T-L$ with exactly two neighbors in $L$.
    By contracting each component of $T-L$ into one of its neighbors in $L$,
    we obtain a tree on $|L|$ nodes, whose edges correspond to $\cal C$; hence $|{\cal C}| \le |L|-1$.
    
    We define $P_0 = \chi(L)$. Clearly, $S \sub P_0$ and $|P_0| \le (\beta+1)\cdot |L| \le 2(\beta+1)\cdot |S|$.
    For each $C \in \cal C$ we create a set $P_C = \chi(C) \sm \chi(L)$ and add it to the decomposition.
    The neighborhood of $P_C$ is contained in two bags of $(T,\chi)$ so $|N_G(P_C)| \le 2(\beta+1)$ and so $N_G[P_C]$ is a $2(\beta+1)$-protrusion.
    Each of the remaining components of $T-L$ has one neighbor in $L$.
    For $t \in L$ we define $C_t \sub V(T)$ as the union of connected components of $T-L$ whose only neighbor in $L$ is $t$.
    Again, we add $P_t = \chi(C_t) \sm \chi(L)$ to the decomposition, for which $N_G[P_t]$ forms a $(\beta+1)$-protrusion.
    The number of sets in the decomposition is $\le 2|L|$ which is at most  $2(\beta+1)\cdot |S|$ because $\beta \ge 1$.
\end{proof}

\lemDichProduct*
\begin{proof}
    Let $\cal M$ be the family of vertex sets corresponding to all the minor models in the $(\fcal, c)$-minor packing.
    Since each $F \in \fcal$ is connected, also each $M \in \cal M$ induces a connected graph.
    Hence for each $j \in [1,q]$ there is at most one connected component of $G[Q_j]$ that contains a minor model from $\cal M$.
    When $G[Q_j]$ has more than one connected component, we split $Q_j$ into $Q^1_j$, containing the above minor model, and $Q^2_j = Q_j \sm Q^1_j$.
    Both $Q^1_j, Q^2_j$ are still $\beta_2$-protrusions in $G-P_0$. 
    We can thus assume that for each $j \in [1,q]$, the set $Q_h$ either induces a connected graph or is disjoint from $\cal M$.
    This increases the number of protrusions by a factor at most 2.

    We now construct protrusion decomposition $(R_0,R_1,\dots,R_r)$ of $G$ by setting $R_0 = P_0 \cup Q_0$ and describing a partition $\cal R$ of $V(G) \sm (P_0 \cup Q_0)$.
    For each $j \in [1,q]$ with $G[Q_j]$ connected, we add $Q_j$ to $\cal R$.
    There are at most $\alpha_2 \le  \alpha_2\cdot\beta_2$ such indices.
    Observe that in this case there must be a unique $i \in [1,p]$ for which $Q_j \sub P_i$ and so $N_G(Q_j) \sub N_G(P_i) \cup N_{G-P_0}(Q_j)$ has size $\le \beta_1 + \beta_2$, hence $Q_j$ forms a $(\beta_1 + \beta_2)$-protrusion.
    This already ensures that the $(\fcal, c)$-minor packing will be preserved.

    Next, we handle disconnected $G[Q_j]$.
    For such an index $j \in [1,q]$, consider the set $I_j \sub [1,p]$ of those indices $i$ for which $N_{G-P_0}(Q_j) \cap P_i \ne \emptyset$.
    Clearly, $|I_j| \le \beta_2$.
    The set $N_G(P_i \cap Q_j)$ is contained in $N_G(P_i) \cup N_{G-P_0}(Q_j)$ which has size at most $\beta_1 + \beta_2$, so again we obtain a $(\beta_1 + \beta_2)$-protrusion.
    For each $i \in I_j$ we insert $Q_j \cap P_i$ to $\cal R$.
    The number of such pairs is $\le \alpha_2 \cdot \beta_2$.
    
    Note that for each $i \in [1,p] \sm I_j$ satisfying $Q_j \cap P_i \ne \emptyset$, the set $Q_j \cap P_i$ must induce a union of connected components of $G[P_i]$.
    In particular, these components must be disjoint from $Q_0$.
    Consequently, for each $i \in [1,p]$ the subset $P'_i \sub P_i \sm Q_0$ that has not been yet covered by $\cal R$ induces a union of connected components of $G[P_i]$ disjoint from $Q_0$.
    We infer that $N_G(P'_i) \sub N_G(P_i)$ so $P'_i$ is a $\beta_1$-protrusion and it can be added to $\cal R$.
    This concerns $\le \alpha_1$ sets and concludes the construction.
\end{proof}

\subsection{Uniform lossy kernelization}\label{app:kernel}

\begin{observation}\label{obs:near-protrusion}
    For any \fcal-deletion set $S$ of $G$ of size at most $k$,
    and any connected component $C$ of $G-(X \cup Z)$,
    $|(N_G(C) \cap X) \setminus S| \leq \eta+1$.
\end{observation}

\begin{proof}[Proof of Observation~\ref{obs:near-protrusion}]
    Let $S$ be an \fcal-deletion set of $G$ of size at most $k$ and let $C$ is an arbitrary connected component of $G-(X \cup Z)$.
    Let $R:=(N_G(C) \cap X) \setminus S$. 
    Since $R\subseteq N_G(C) \cap X$, from Lemma~\ref{lem:near-protrusion}, for each $u,v \in R$, there exists at least $k+\eta+2$ vertex-disjoint $u$ to $v$ paths in $G$. Since $|S| \leq k$ and $R\cap S =\emptyset$, there exists at least $\eta+2$ vertex-disjoint $u$ to $v$ paths in $G-S$, for each pair $u,v \in R$. Since the treewidth of $G-S$ is at most $\eta$, $R$ is contained inside some bag of every tree decomposition of $G-S$ of width at most $\eta$. Therefore, $|R| \leq \eta +1$.
\end{proof}

\begin{proof}[Proof of \Cref{lem:lrr-one-correct}]
    Clearly both the reduction algorithm and the solution lifting algorithm run in polynomial time.
    Let $I:=(G,k)$ and $I':=(G',k')$.
    To prove the lemma, it remains to prove that

\begin{equation}\label{eqn:one}
    \frac{\val^{\fdelsmall}_{I}(S)}{\opt_{\fdelsmall}(I)} \leq \max \left \{(1+\epsilon), \frac{\val^{\fdelsmall}_{I'}(S')}{\opt_{\fdelsmall}(I')} \right\}.
\end{equation} 
If $S'$ is not an \fcal-deletion set of $G'$ then \Cref{eqn:one} holds.
Now,  let $S'$ be an \fcal-deletion set of $G'$.
Then $S$ is also an \fcal-deletion set of $G$. 
Also if $\val^{\fdelsmall}_{I'}(S') = k'+1$ then 
$\val^{\fdelsmall}_{I}(S) =k+1$ and \Cref{eqn:one} holds.
Otherwise, $\val^{\fdelsmall}_I(S) \leq \val^{\fdelsmall}_{I'}(S') + |N_G(C) \cap X|$.
Let $O$ be an optimal solution in $G$. From \Cref{obs:near-protrusion}, $|O \cap (N_G(C) \cap X)| \geq |N_G(C) \cap X| - (\eta +1)$. Therefore, $\opt_{\fdelsmall}(I) \geq \opt_{\fdelsmall}(I') + |N_G(C) \cap X| - (\eta+1)$.
Thus, we get that,

\begin{align*}
    \frac{\val^{\fdelsmall}_I(S)}{\opt_{\fdelsmall}(I)} &\leq \frac{\val^{\fdelsmall}_{I'}(S')+ |N_G(C) \cap X|}{\opt_{\fdelsmall}(I') + |N(C) \cap X| -(\eta+1)} &\\
                             &\leq \max \left\{\frac{\val^{\fdelsmall}_{I'}(S')}{\opt_{\fdelsmall}(I')}, \frac{|N_G(C) \cap X|}{|N_G(C) \cap X| -(\eta+1)} \right \} & \text{ (from Proposition~\ref{prop:fact})}\\
                             & \leq \max \left\{ \frac{\val^{\fdelsmall}_{I'}(S')}{\opt_{\fdelsmall}(I')},  (1+\epsilon)\right \} &
\end{align*}
The last inequality above follows because if, for the sake of contradiction, $\frac{|N_G(C) \cap X|}{|N_G(C) \cap X|-(\eta+1)} > (1+\epsilon)$, then it implies that $|N_G(C) \cap X| < \frac{(1+\epsilon)(\eta+1)}{\epsilon}$, which is a contradiction.
\end{proof}

\begin{observation}\label{lem:bound-neighbourhood}
    When \Cref{lrr:protrusion} is no longer applicable, for each connected component of $C$ of $G-(X \cup Z)$, it holds that
    $|N(C)| \leq \eta'$
    where $\eta':= \frac{(1+\epsilon)(\eta+1)}{\epsilon} + 2(\eta+1)$.
\end{observation}

The proof of Lemma~\ref{lem:protrusion-decomposition} follows from \Cref{lem:bound-neighbourhood}.

\begin{proof}[Proof of \Cref{lem:protrusion-decomposition}]
   Define $\eta^*:= \eta + \eta'$, where $\eta'$ is as defined in Lemma~\ref{lem:protrusion-decomposition}. Then $\eta^*=\Oh(\eta / \epsilon)$.
   The construction of the desired $(\alpha, \eta^*)$-protrusion decomposition of $V(G) = (P_0, \ldots, P_{\ell})$ is pretty straightforward.
   The set $P_0:=X \cup Z$. From Lemma~\ref{lem:near-protrusion}, $|P_0| = \Oh_{\eta}(k^3)$.
   Partition the set of connected components of $G-(X \cup Z)$ into parts $(\mathcal{C}_1, \ldots, \mathcal{C}_{\ell})$ such that for each $i\in [1, \ell]$, 
   the connected components in part $\mathcal{C}_i$ have the same neighbourhood, and for any $i,j \in [1,\ell]$, $i\neq j$, 
   $C \in \mathcal{C}_i$ and $C' \in \mathcal{C}_j$,
   $N(C) \neq N(C')$.
   For each $i \in [1,\ell]$,
   set $P_i:= \bigcup_{C \in \mathcal{C}_i} C$.
   It is straightforward to verify that each $P_i$ satisfies the properties stated in the lemma.
   
    Finally, we show that $N[P_i]$ is an $\eta^*$-protrusion. Indeed,
    $|N(P_i)|= |N[\bigcup_{C \in \mathcal{C}_i} C]| = |N(C^*)|$ for any $C^* \in \mathcal{C}_i$ and $|N(C^*)| \leq \eta'$ from \Cref{lem:bound-neighbourhood}.
    Furthermore, from Lemma~\ref{lem:near-protrusion}, $X$ is an \fcal-deletion set of $G$. 
    Therefore for each connected component $C$ of $G-(X \cup Z)$,
    $G[C]$ has no graph from $\mathcal{F}$ as a minor, and since $\mathcal{F}$ contains the planar graph $F_{\planar}$, the treewidth of $G[C]$ is at most $\eta$. 
    Therefore, the treewidth of $G[\bigcup_{C \in \mathcal{C}_i} C]$ is also upper bounded by $\eta$.
    Finally, the treewidth of $G[N[\bigcup_{C \in \mathcal{C}_i} C]]$ is at most the treewidth of $G[\bigcup_{C \in \mathcal{C}_i} C]$, plus $|N(\bigcup_{C \in \mathcal{C}_i} C)|$, which is at most $\eta + \eta'$.
    Therefore, $N[P_i]$ is an $\eta^*$-protrusion in $G$.
\end{proof}

\begin{proof}[Proof of \Cref{lem:fdel-tweta-flow}]
    Consider the \twetadel\ problem first. 
    Since $G_{\flow}$ is a supergraph of $G$, any solution of $G_{\flow}$ is also a solution of $G$.
    For the other direction, let $S \subseteq V(G)$ such that $|S| \leq k$ and the treewidth of $G-S$ is at most $\eta$.  Let $(T,\beta)$ be a tree decomposition of $G-S$ of width at most $\eta$.
    We show that $(T,\beta)$ is also a tree decomposition of  $G_{\flow}-S$. 
    For this, we show that for each edge $uv \in E(G_{\flow}) \setminus E(G)$,
    either $S \cap \{u,v\} \neq \emptyset$, or there exists $t \in V(T)$ such that $u,v \in \beta(t)$.

    If $S \cap \{u,v\} = \emptyset$, then in $G -S$, there are at least $\eta+2$ vertex-disjoint paths between $u$ and $v$.
    For the sake of contradiction, say it is not the case  that there exists $t \in V(T)$ such that $u,v \in \beta(t)$. Root the tree $T$ arbitrarily. Let $t_u \in V(T)$ (resp. $t_v \in V(T)$) be the closest to the root of $T$ such that $u \in \beta(t_u)$ and $v \in \beta(t_v)$.
    By assumption, $t_u \neq t_v$.
    Let $t \in V(T)$ be the neighbour of $t_u$ on the unique path between $t_u$ and $t_v$ in $T$.
    Then $(\beta(t_u) \cap \beta(t)) \cap \{u,v\} = \emptyset$. Also, $\beta(t_u) \cap \beta(t)$ is a $u$-$v$ separator in $G-S$. Since $|\beta(t_u) \cap \beta(t)| \leq \eta +1$, by the Menger's theorem there are at most $\eta +1$ vertex-disjoint paths between $u$ and $v$ in $G-S$, which is a contradiction.

    Next we move to the \fdel\ problem. Let $S$ be an \fcal-deletion set of $G$ of size at most $k$. We show that the treewidth of $G_{\flow} -S$ is at most $\eta$. Since \fcal\ contains the planar graph $F_{\planar}$, treewidth of $G-S$ is at most $\eta$. From the above arguments, if $(T,\beta)$ is a tree decomposition of $G-S$ of width at most $\eta$, then $(T,\beta)$ is also a tree decomposition of $G_{\flow} - S$.
\end{proof}

\begin{proof}[Proof of \Cref{lem:non-clique}]
    From the construction of $G_{\flow}$,
    for each $u,v \in P_0$ such that $uv \not \in E({G_{\flow}})$,
    there exists at most $k+\eta+1$ vertex-disjoint paths from $u$ to $v$ in $G$.
    For each $P_i$, $i \geq 1$, if $u,v \in N(P_i)$, then there exists a path from $u$ to $v$ in $G[P_i \cup \{u,v\}]$.

    Since for each $i,j \geq 1$, $i \neq j$, $P_i \cap P_j =\emptyset$, we conclude that the size of 
    $\mathcal{P}_{u,v}: =\{P_i: P_i \in \mathcal{P}_{\nonclique}, u,v \in N(P_i)\}$ is at most $ k+\eta+1$, as otherwise the edge $uv$ would be present in $G_{\flow}$.
    Since the neighbourhood of every component in  $\mathcal{P}_{\nonclique}$ is not a clique, 
    for every component of  $\mathcal{P}_{\nonclique}$,
    there exists a non-adjacent pair in $G_{\flow}[P_0]$ which is in the neighbourhood of this component.

    Since the number of non-adjacent pairs in $G_{\flow}[P_0]$ can be naively bounded by ${|P_0| \choose 2} = \Oh(|P_0|^2) = \Oh_{\eta}(k^6)$ (Lemma~\ref{lem:protrusion-decomposition}), 
    this bounds the size of $\mathcal{P}_{\nonclique}$ by $\Oh_{\eta}(k^7)$. Below we give a better bound to prove the lemma.

    Recall that $P_0 = X \cup Z$.
    Also $|X| = \Oh_{\eta}(k)$ and $|Z|= \Oh_{\eta}(k^3)$ from Lemma~\ref{lem:near-protrusion}.
    Furthermore, since from Lemma~\ref{lem:near-protrusion}, since $X$ is an \fdel\ set of $G$ and $X \cap Z =\emptyset$, 
     and the treewidth of $G_{\flow}[Z]$ is at most $\eta$.

     The number of components in $\mathcal{P}_{\nonclique}$ that have a pair  of non-adjacent neighbours in $X$ is at most at most $(k+\eta+1) \cdot {|X| \choose 2} = \Oh_{\eta}(k^3)$.
     Similarly, the number of components in $\mathcal{P}_{\nonclique}$ that have a pair of non-adjacent neighbours, one of which is in $X$ and the other is in $Z$ is at most $(k+\eta+1) \cdot (|X| \cdot |Z|) =\Oh_{\eta}(k^5)$.

     It remains to bound the number of components in $\mathcal{P}_{\nonclique}$ that have a pair of non-adjacent neighbours in $Z$.
     With each such component $C$, associate an arbitrary {\em non-edge witness pair with both endpoints in $Z$}, that is $\{x,y\}$, where $x,y \in Z \cap N(C)$ such that $xy \not \in E(G_{\flow})$.
     We can then bound the number of non-edge witnesses with both endpoints in $Z$, by the number of edges in a graph of treewidth $\eta$ on vertex set $|Z|$ as follows. 
     Observe that for a non-simplicial component $C$ with a witness pair $\{u,v\}$,
     we can contract $C$ into $u$ to obtain the edge $uv$, and this contraction does not increase the treewidth of $G-X$. 
     Hence each non-simplicial component that sees a non-edge with both endpoints in $Z$, can be charged to an edge of this bounded treewidth (bounded by at most $\eta$) minor on vertex set $Z$.
     Since the number of edges on any $n$-vertex graph of treewidth at most is $\eta$ is $\Oh_{\eta}(n)$ (because treewidth at most $\eta$ graphs have degeneracy at most $\eta$), the number of edges of this minor on the vertex set $Z$, and hence the number of non-edge witness pairs in $Z$, are $\Oh_{\eta}(|Z|) = \Oh_{\eta}(k^3)$. This bounds the number of components in $\mathcal{P}_{\nonclique}$ that have a non-edge witness pair in $Z$ by $\Oh_{\eta}(k^4)$.

     This concludes that the size of $\mathcal{P}_{\nonclique}$ is $\Oh_{\eta}(k^5)$.
\end{proof}

\begin{proof}[Proof of \Cref{lem:lrr-two-correct}]
    It is clear that the reduction algorithm runs in polynomial time. 
    We first show that the solution lifting algorithm runs in polynomial time.
    
    \begin{claim}
        If $S'$ is a solution of \twetadel\ on the instance $(G',k)$,
    then the
    treewidth of $G-S'$ is $\Oh(\eta)$.
    \end{claim}
    
    \begin{proof}
    The treewidth of $G'-S'$ is at most $\eta$.
    Let $(T,\beta)$ be a tree decomposition of $G'-S'$ of width at most $\eta$.
    For each $P_i \in \mathcal{P}_{\clique}$, by the definition of $\mathcal{P}_{\clique}$, $N_{G_{\flow}}(P_i)$ is a clique.
    Therefore, 
    $N_{G_{\flow}}(P_i) \setminus S'$ is also a clique and it is contained in some bag of $(T,\beta)$ (\cite[Lemma~$12.3.5$]{DBLP:books/daglib/0030488}).
    Moreover, $|N_{G_{\flow}}(P_i) \setminus S'| \leq \eta +1$. 

    From Lemma~\ref{lem:protrusion-decomposition}, $P_i$ is a disjoint union of connected components of $G-(X \cup Z)$ where $X$ is an \fcal-deletion set, therefore
   the treewidth of $G[P_i]=G_{\flow}[P_i]$ is at most $\eta$.
   Since $|N_{G_{\flow}}(P_i) \setminus S'| \leq \eta +1$,
   a tree decomposition of $G_{\flow}[N[P_i]] \setminus S'$ of width  at most $2 \eta +1$ can be obtained from a tree decomposition of width at most $\eta$ for $G_{\flow}[P_i]$ by adding $N(P_i) \setminus S'$ to each of its bags. 
   Let such a tree decomposition of $G_{\flow}[N[P_i]] \setminus S'$ be $(T_i,\beta_i)$.

   A tree decomposition of $G_{\flow}-S'$ of width $\Oh(\eta)$ can be obtained from $(T,\beta)$, by doing the following procedure, for each $P_i \in \mathcal{P}_{\clique}$: 
   let $t_i \in V(T)$ be a bag of $(T,\beta)$ containing $N_{G_{\flow}}(P_i) \setminus S'$. 
   Attach the tree $T_i$ to $T$ by making an arbitrary node $t \in V(T_i)$ adjacent to $t_i \in V(T)$. 
   It can be seen easily that the resulting tree represents a tree decomposition of $G_{\flow}-S'$  of width which is at most the maximum of the width of $(T,\beta)$ and $(T_i,\beta_i)$ for all $i$ such that $P_i \in \mathcal{P}_{\clique}$.
   Finally, since $G$ is a subgraph of $G_{\flow}$, the treewidth of $G-S'$ is at most the treewidth of $G_{\flow}-S'$. 
   \end{proof}

    The solution lifting algorithm runs the algorithm from \Cref{lem:f-del-on-tw} on the graph $G-S'$, to compute an optimum \fcal-deletion set $S''$ of $G-S'$.
    Since treewidth of $G-S'$ is $\Oh(\eta)$ from the above claim, the solution lifting algorithm runs in time $\Oh_{\fcal,\eta}(n)$.

    Let $I:=(G,k)$ and $I':=(G',k)$.
    If $S'$ is a solution of \twetadel\ on the instance $I'$, then $S:= S' \cup S''$ is clearly an \fcal-deletion set of $G$ because $S''$ is an \fcal-deletion set of $G-S'$.
    If $\val^{\twetadelsmall}_{I'}(S') = k'+1$,
    then $\val^{\fdelsmall}_{I}(S) = k+1$. 
    From Lemma~\ref{lem:fdel-tweta-flow},
     $\opt_{\twetadelsmall}(I') \leq \opt_{\fdelsmall}(I)$. Therefore, in this case we already obtain 
      \begin{align*}
       \frac{\val^{\fdelsmall}_I(S)}{\opt_{\fdelsmall}(I)} \leq  \frac{\val^{\twetadelsmall}_{I'}(S')}{\opt_{\twetadelsmall}(I')}.
   \end{align*}

    Finally,
    to prove the lemma, it remains to show that if $S'$ is a solution of size at most $k$ of \twetadel\ on the instance $I'$, then

    \begin{align*}
       \frac{|S|}{\opt_{\fdelsmall}(I)} \leq 2\cdot  \frac{|S'|}{\opt_{\twetadelsmall}(I')}.
   \end{align*}

Since $|S''|= \opt_{\fdelsmall}((G-S',k))$ and $|S|=|S'|+|S''|$, we get

        \begin{align*} 
        \frac{|S|}{\opt_{\fdelsmall}(I)}
                    &= \frac{|S'|}{\opt_{\fdelsmall}(I)} +  \frac{\opt_{\fdelsmall}((G-S',k))}{\opt_{\fdelsmall}(I)} 
                     \leq  \frac{|S'|} {\opt_{\fdelsmall}(I)} +  1  &  \\
                     & \leq  \frac{|S'|} {\opt_{\twetadelsmall}((G_{\flow},k))} +  1 & \text{ (from Lemma~\ref{lem:fdel-tweta-flow})}\\
                     & \leq  \frac{|S'|} {\opt_{\twetadelsmall}(I')} +  1
                     \leq  2\cdot \frac{|S'|} {\opt_{\twetadelsmall}(I')}. & 
    \end{align*}
\end{proof}

\begin{figure}[h!]
    \includegraphics[scale=0.5]{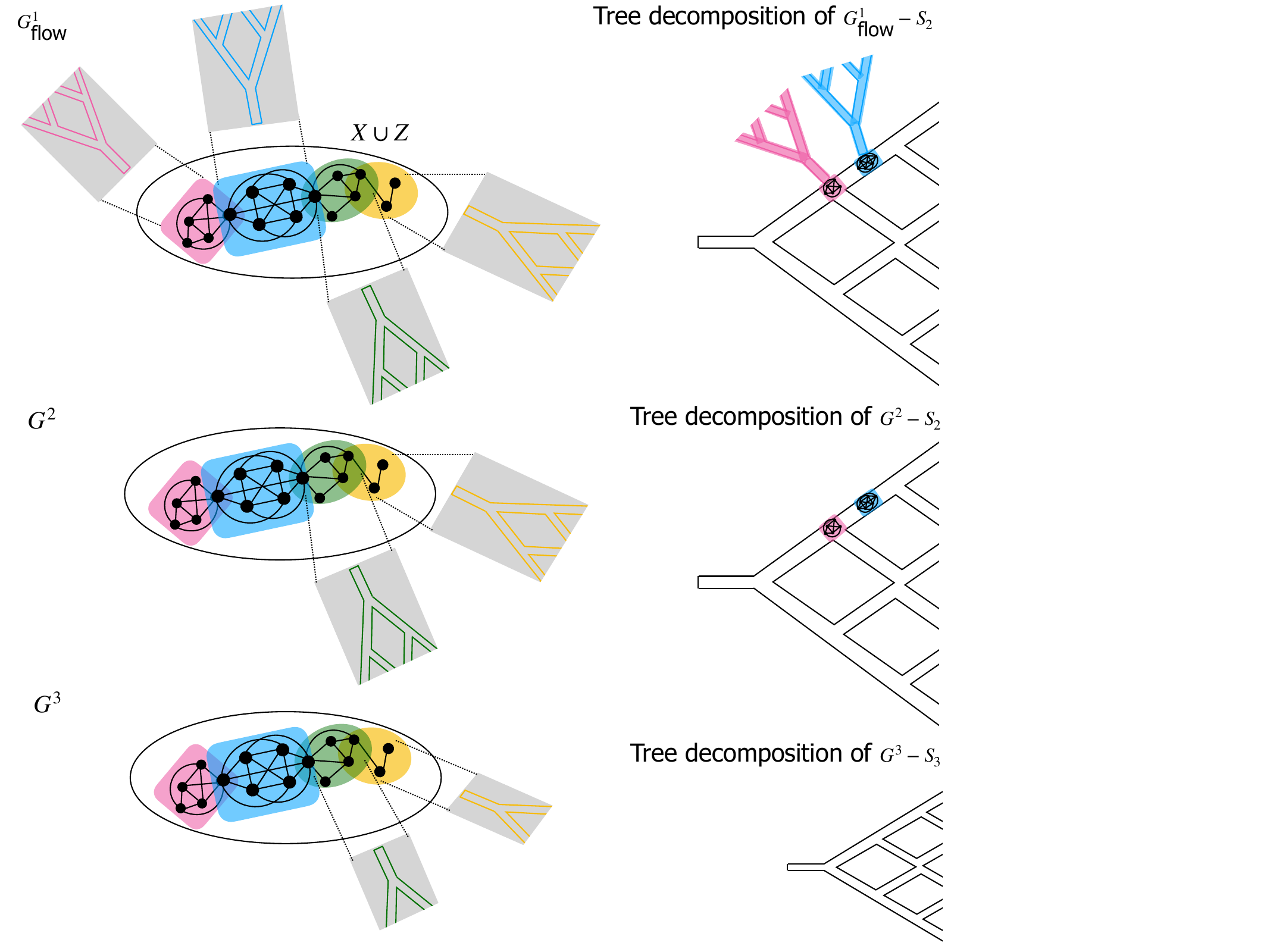}
    \caption{Illustration to \Cref{thm:lossy-kernel}. Left column: Top figure represents the protrusion decomposition; the oval is the set $X \cup Z$ and the different colored parts are $P_1,P_2,P_3,P_4$; the pink and blue parts are simplicial with pink and blue clique neighbourhoods in $X \cup Z$ respectively; the green and yellow parts are non-simplicial with green and yellow parts in $X \cup Z$ as their neighbourhoods respectively. The middle graph is obtained after removing the simplicial parts (pink and blue) and the last graph is obtained after reducing the simplicial parts into constant size using \Cref{lem:dichotomy:compress}. \\
    Right column: The bottom-most part represents the tree decomposition of $G^3-S_3$, the middle part is the tree decomposition of $G^2-S_2$ and the top part is the tree decomposition of $G^1_{\flow} - S_2$ obtained by attaching the pink and blue tree decompositions of the pink and blue parts, after adding their pink and blue neighbourhoods to each bag of their tree decompositions, respectively.}
    \label{fig:lossy-kernel}
\end{figure}

\subsection{Uniform lossy protocol}\label{app:protocol}

\begin{proof}[Proof for \Cref{lem:lrr-three}]
    Clearly, if $S'$ is an \fcal-deletion set of $G'$ then $S$ is an \fcal-deletion set of $G$.
    Let $I:=(G,k)$ and $I':=(G',k')$.
    If $|S'| > k'$, then $|S| >k$, hence, $\val^{\fdelsmall}_I(S) = \val^{\fdelsmall}_{I'}(S')$ and the lemma follows.
    Finally,
    say $S'$ is an \fdel\ of $G'$ of size at most $k'$. 
    It remains to prove that 

    $$\frac{|S|}{\opt_{\fdelsmall}(I)} \leq (1+\epsilon) \cdot \beta.$$

    From Lemma~\ref{lem:partition-mod} two cases can arise: either $|Q_0| \leq (1+\epsilon)\beta \cdot \opt_{\twetadelsmall}((G_{\flow}[Q_0],k))$ or $|Q_0| \leq \epsilon \beta \cdot \opt_{\twetadelsmall}((G_{\flow}[P_0],k))$. 

    In the first case, 
    $|S| \leq |S'| + |Q_0| 
            \leq |S'| + (1+\epsilon)\beta \cdot \opt_{\twetadelsmall}((G_{\flow}[Q_0],k))$,
and 
        $\opt_{\fdelsmall}(I) \geq \opt_{\fdelsmall}(I') + \opt_{\fdelsmall}((G[Q_0],k)) 
            \geq \opt_{\fdelsmall}(I') + \opt_{\twetadelsmall}((G_{\flow}[Q_0],k))$ from Lemma~\ref{lem:fdel-tweta-flow}.
 Combining these we get

    \begin{equation*}
        \frac{|S|}{\opt_{\fdelsmall}(I)} \leq \frac{|S'| + (1+\epsilon)\beta \cdot \opt_{\twetadelsmall}((G_{\flow}[Q_0],k))}{\opt_{\fdelsmall}(I') + \opt_{\twetadelsmall}((G_{\flow}[Q_0], k))} 
        \leq \max  ( \beta, {(1+\epsilon)\beta }). 
    \end{equation*}
    where the latter inequality follows from Proposition~\ref{prop:fact}.
    In the second case, 
    \begin{align*}
        \frac{|S|}{\opt_{\fdelsmall}(I)} &= \frac{|S'|}{\opt_{\fdelsmall}(I)} + \frac{\epsilon \beta \cdot \opt_{\twetadelsmall}((G[P_0],k))}{\opt_{\fdelsmall}(I)} & \\
        &  \leq \frac{|S'|}{\opt_{\fdelsmall}(I')} + \frac{\epsilon \beta \cdot \opt_{\fdelsmall}((G[P_0],k))}{\opt_{\fdelsmall}(I)} & \text{ (from Lemma~\ref{lem:fdel-tweta-flow})}\\
        &  \leq \beta + \epsilon \beta = (1+\epsilon)\beta. &
    \end{align*}
\end{proof}

\subsection{Computable protrusion replacer: Proof of Lemma~\ref{lem:prelim:replace}}\label{sec:computable-protrusion-repalcer}

\protrusionreplacer*

\noindent{\bf{Encoder.}}
Let $\mathcal{F}_p$ be the collection of all graphs on at most $p$ vertices, where $p:= \max_{F \in \mathcal{F}} |V(F) |   +   r$.
Let $\mathcal{C}$ be the collection containing sets of pairs $(F,\mathcal{A})$, where $F \in \mathcal{F}_p$ and $\mathcal{A}$ is a partition of $I$ where $I \subseteq [1,r]$.
Observe that $|\mathcal{C}| \leq 2^{2^{\Oh(r \log r)} 2^{\Oh(p^2)} \cdot 2^r}$ (because $|\mathcal{F}_p| =  2^{\Oh(p^2)}$).

Given any $r$-boundaried graph $H^*$,
$F \in \mathcal{F}_p$, $B \subseteq \partial(H)$,
and a partition $\mathcal{A}$ of $B$,
we say that {\em $(F,\mathcal{A})$-rooted minor is witnessed by $H^*$},
if there exists a minor model $\psi$ of $F \in \mathcal{F}_p$ in $H^*$, such that each part of the partition of $\mathcal{A}$ is a subset of some branch set of the minor model of $\psi$, and no vertex of $\partial(H^*) \setminus \bigcup_{A \in \mathcal{A}} A$ is in any branch set of $\psi$.

Define $\mathcal{L}_{\mathcal{C}}$ as the language containing triplets $(H^*,S,R)$, where $H^*$ is an $r$-boundaried graph, $S \subseteq V(H^*)$,
and $R \in \mathcal{C}$,
such that $(F,\mathcal{A}) \in R$ if and only if $(F,\mathcal{A})$ is witnessed as a rooted minor in $H^*-S$.
The pair $\mathcal{E}:=(\mathcal{C},\mathcal{L}_{\mathcal{C}})$ is called the {\em encoder for \fdel}.

For any $R \in \mathcal{C}$, let $f_{H^*}(R) := \min \{d: \exists S^* \subseteq V(H^*), |S^*| \leq d, (G,S,R) \in \mathcal{L}_{\mathcal{C}}\}$. If such a set $S^*$ does not exist, then $f_{H^*}(R) = +\infty$.

\noindent{\bf Equivalence relation.} 
For any two $r$-boundaried graphs $H_1,H_2$ of treewidth at most $r$, we say that $H_1 \sim^r_{\mathcal{E}} H_2$ if
\begin{itemize}
    \item for each $R \in \mathcal{C}$,
    $f_{H_1}(R) = f_{H_2}(R)$, and

    \item $h$-{\em folio}$(H_1 - \partial(H_1)) = h$-{\em folio}$(H_2  -   \partial(H_2))$. 
\end{itemize}

For any graph class $\mathcal{G}$, we say that $H_1 \sim^r_{\mathcal{G}} H_2$, 
if 
for each $r$-boundaried graph $G^*$, $(H_1 \oplus G^*) \in \mathcal{G}$ if and only $(H_2 \oplus G^*) \in \mathcal{G}$.
We say that $H_1 \sim^r_{\mathcal{E}, \mathcal{G}} H_2$, 
if $H_1 \sim_{\mathcal{E}}  H_2$ and
for each $r$-boundaried graph $G^*$, $(H_1 \oplus G^*) \in \mathcal{G}$ if and only $(H_2 \oplus G^*) \in \mathcal{G}$.

\begin{lemma}\label{lem:fdel-equiv}
    Let $G$ be a graph such that $G=(G^* \oplus H)$, where $H$ is an $r$-boundaried graph.
    Let $\widehat{H}$ be an $r$-boundaried graph such that $H \sim^r_{\mathcal{E}}    \widehat{H}$ (or $H \sim^r_{\mathcal{E}}    \widehat{H}$).
    Let $G'=(G^* \oplus \widehat{H})$.

    If $G$ has an \fdel\ set $S$ such that $|S \cap V(H)| \leq d$, then $G'$ has an \fdel\ set $S'$ such that $|S'| = |S|$ and $|S' \cap V(\widehat{H})| \leq d$. Such a set can be computed in linear time.
    If $G'$ has an \fdel\ set $S'$ such that $|S' \cap \widehat{H}| \leq d$, then there exists an \fdel\ set $S$ of $G$ such that $|S| = |S'|$ and $|S \cap V(H)| \leq d$. Such a set can also be computed in linear time.
\end{lemma}
\begin{proof}
    Let $S$ be an \fdel\ set of $G$ such that $|S \cap V(H)| \leq d$.
   Let $X_H:= S \cap V(H)$.
    Let $R$ be the collection of those pairs $(F,\mathcal{A})$, $F \in \mathcal{F}_p$ and $\mathcal{A}$ is a partition of some subset of $\partial(H)$, such that $(F,\mathcal{A})$ is witnessed as a rooted minor in $H - X_H$.

    Since $H \sim^r_{\mathcal{E}} \widehat{H}$,
    $f_H(R) = f_{\widehat{H}}(R)$.
    That is, there exists $X_{\widehat{H}} \subseteq V(\widehat{H})$ such that $|X_{\widehat{H}}| = |X_{H}|$ and $\widehat{H} - X_{\widehat{H}}$ witnesses exactly the pairs in $R$ as rooted minors.
    Since $H$ has bounded treewidth, such a set $X_H$ can also be computed in linear time using Courcelle's theorem.

    We now claim that $S':=(S \setminus X_H) \cup X_{\widehat{H}}$ is an \fdel\ set in $G'$.
    For the sake of contradiction, say there exists $F^* \in \mathcal{F}$ such that $F^*$ is a minor of $G' - S'$.
    Let $\psi^*$ be a minor model of $F^*$ in $G'-S'$.
    Consider the restriction of $\psi^*$ on the subgraph $H$ of $G$. In this restriction, some branch sets of $\psi$ may not be connected. But, since $H$ is an $r$-boundaried graph, this restriction has at most $|V(F^*)| +r$ branch sets and therefore, witnessing a $(F,\mathcal{A})$-rooted minor model where $F$ is a graph on at most $p$ vertices.
    Since $(F, \mathcal{A})$ is also witnessed as a rooted minor by $H-X_H$, this implies that $F$ is a minor of $G-S$ as well, which is a contradiction.
    The other direction is analogous.
\end{proof}

\begin{lemma}\label{lem:equiv-one}
    The number of equivalence classes of $\sim^r_{\mathcal{E}}$ is at most $(d+2)^{2^{2^{\Oh(r \log r) }    2^{\Oh(p^2)} \cdot 2^r}} \cdot 2^{2^{\Oh(p^2)}}$.
\end{lemma}
\begin{proof}
    For any $r$-boundaried graph $H^*$, the number of different functions $f_{H^*}$ is at most $(d+2)^{|\mathcal{C}|}$.
    The number of graphs of at most $p$ vertices is $2^{\Oh(p^2)}$ and so the number of different collections containing graphs on at most $h$ vertices, and thus the  number of distinct $h$-{\em folios}, is at most $2^{2^{\Oh(p^2)}}$. 
\end{proof}

\begin{proposition}[\cite{DBLP:journals/siamdm/GarneroPST15}]\label{prop:equiv-top-minor}
    If $\mathcal{G}$ is a class of graphs that exclude a fixed graph $Q$ as a (topological) minor,
    then the number of equivalence classes of $\sim^r_{\mathcal{G}}$ is at most $2^{r \log r} \cdot |V(Q)|^{r} \cdot 2^{|V(Q)|^2}$.
\end{proposition}

\begin{lemma}\label{lem:equiv-two}
    The number of equivalence classes of $\sim_{\mathcal{E}, \mathcal{G}}$, when $\mathcal{G}$ is the class of graphs that exclude a fixed graph $Q$ as a (topological) minor is at most 
    $(d+2)^{2^{2^{\Oh(r \log r) }    2^{\Oh(p^2)} \cdot 2^r}} \cdot 2^{2^{\Oh(p^2)}} \cdot 2^{r \log r} \cdot |V(Q)|^{r} \cdot 2^{|V(Q)|^2}$.
\end{lemma}
\begin{proof}
    This follows from Lemma~\ref{lem:equiv-one} and Proposition~\ref{prop:equiv-top-minor}.
\end{proof}

Henceforth, let $\texttt{index}_{\mathcal{E}}:= (d+2)^{2^{2^{\Oh(r \log r) }    2^{\Oh(p^2)} \cdot 2^r}} \cdot 2^{2^{\Oh(p^2)}}$ 
and $\texttt{index}_{\mathcal{E}, \mathcal{G}}: = (d+2)^{2^{2^{\Oh(r \log r) }    2^{\Oh(p^2)} \cdot 2^r}} \cdot 2^{2^{\Oh(p^2)}} \cdot 2^{r \log r} \cdot |V(Q)|^{r} \cdot 2^{|V(Q)|^2}$.

\noindent{\bf DP-friendly equivalence relation.} 
Let $H$ be an $r$-boundaried graph of treewidth at most $r$.
Let $(T, \beta)$ be a rooted tree decomposition 
of $H$ such that $\partial(H)$ is in the root bag of $(T,\beta)$.
For any node $t$ of the tree decomposition $(T,\beta)$,
let $H_t$ be the $r'$-boundaried graph $H$ induced by the vertices in all the bags that are descendants of $t$ in $T$, including $t$, with $\beta(t)$ as the boundary.

We say that $\sim^r_{\mathcal{E}}$ (and $\sim^r_{\mathcal{E}, \mathcal{G}}$)
is {\em DP-friendly} 
if for any $t \in V(T)$, $r' \leq r$,
for any $r'$-boundaried graph $H^*$ such that $H_t \sim^{r'}_{\mathcal{E}} H^*$ (and $H_t \sim^{r'}_{\mathcal{E}, \mathcal{G}} H^*$),
it holds that $H \sim^r_{\mathcal{E}} H'$ (resp.~$H \sim^r_{\mathcal{E}, \mathcal{G}} H'$), where $H'$ is obtained from $H$ by replacing $H_t$ with $H^*$. That is, $H' = (H \setminus H_t) \oplus H^*$.

\begin{lemma}\label{lem:dp-friendly}
    The equivalence relation $\sim^r_{\mathcal{E}}$ (and $\sim^r_{\mathcal{E}, \mathcal{G}}$)
    is DP-friendly.
\end{lemma}
\begin{proof}
    We prove the lemma for $\sim^r_{\mathcal{E}}$.
    The proof for $\sim^r_{\mathcal{E}, \mathcal{G}}$ is analogous.

    Let $H$ be an $r$-boundaried graph of treewidth at most $r$. 
    Let $(T,\beta)$ be a rooted tree decomposition of $H$ such that $\partial(H)$ is contained in the root bag of $(T,\beta)$.
    For any $t \in V(T)$, 
    let $H_t$ be the boundaried graph $H$ induced by the vertices which are contained in the bags of $(T,\beta)$ which are descendants of $t$, including $t$, where the boundary is $\beta(t)$.
    Let $H^* \sim^r_{\mathcal{E}} H_t$. Let $H':=(H\setminus H_t) \oplus H^*$.
    
    We want to show that $H' \sim^r_{\mathcal{E}} H$.
    We will show that for any $R \in \mathcal{C}$, $f_{H'}(R) = f_{H}(R)$.
    Let $X_H \subseteq V(H)$ such that $|X_H|$ is minimum and $(H,X_H,R) \in \mathcal{L}_{\mathcal{C}}$.
    Let $X_{H_t}:=X_H \cap V(H_t)$.
    Let $R_{H_t}$ be the set of those $(F,\mathcal{A})$ pairs which are witnessed as rooted minors in $H_t \setminus R_{H_t}$. 
    Since $H_t \sim^r_{\mathcal{E}}  H^*$,
    there exists a set $X_{H^*} \subseteq V(H^*)$ such that $|X_{H^*}| = |X_{H_t}|$ and $H^*-X_{H^*}$ witnessing exactly the pairs in $R_{H_t}$ as rooted minors.
    Then it is easy to see that if $S \subseteq V(H)$ such that $H-S$ witnesses exactly the pairs in $R$ as rooted minors, then $S':= (S \setminus X_{H_t}) \oplus X_{H^*}$ is a subset of $V(H')$ and $H'-S'$ witnesses exactly the pairs in $R$ as rooted minors. Also $|S'| = |S|$. Thus, $f_{H}(R) = f_{H'}(R)$.

    We now show that $h$-{\em folio}$(H -  \partial(H)) = h$-{\em folio}$(H'  - \partial(H'))$.
    Let $R$ be the collection of all pairs $(F, \emptyset)$, such that $F \in h$-{\em folio}$(H - \partial(H))$.
    Then $(H,\emptyset,R) \in \mathcal{L}_{\mathcal{C}}$. Because $H \sim^r_{\mathcal{E}}   H'$, we conclude that $h$-{\em folio}$(H - \partial(H)) = h$-{\em folio}$(H' -\partial(H'))$.
\end{proof}

\noindent{\bf Representative size.}

\begin{lemma}\label{lem:rep-size}
    If $\sim^r_{\mathcal{E}}$ (or $\sim^r_{\mathcal{E}, \mathcal{G}}$) is DP-friendly, 
    then each equivalence class of this equivalence relation contains a graph with at most $2^{\texttt{index}_{\mathcal{E}}  +1 } \cdot (r+1)$ (resp. $2^{\texttt{index}_{\mathcal{E}, \mathcal{C}}  +1 } \cdot (r+1)$ ).
\end{lemma}
\begin{proof}
    We prove the lemma for $\sim^r_{\mathcal{E}}$.
    The proof for $\sim^r_{\mathcal{E}, \mathcal{G}}$ is similar.
    Fix an equivalence class $\mathfrak{C}$ of $\sim^r_{\mathcal{E}}$.
    Let $H' \in \mathfrak{C}$ be a graph with the smallest number of vertices.
    We will show that $|V(H)| \leq 2^{\texttt{index}_{\mathcal{E}}  +1 } \cdot r$.

    Let $(T,\beta)$ be a rooted nice tree decomposition of $H'$ of width at most $r$ such that $\partial(H')$ is contained in the root bag.
    For any $t \in V(T)$, by ${H'}_t$ we denote the boundaried graph $H$ induced on the vertices which are in the bags which are descendants of $t$ in $T$, including $t$, with $\beta(t)$ as its boundary.

    Consider a root to some leaf path in $T$. Let $s,t \in V(T)$ be two nodes on such a path.
    We first claim that ${H'}_s \not  \sim^r_{\mathcal{E}} {H'}_t$.
    Indeed, as otherwise say without loss of generality ${H'}_s  \subsetneq {H'}_t$.
    Because $\sim^r_{\mathcal{E}}$ is DP-friendly
    $((H' \setminus {H'}_t) \oplus {H'}_s) \sim^r_{\mathcal{E}}  H'$.
    Because $|V((H' \setminus {H'}_t) \oplus {H'}_s))| < V(H')$,
    this contradicts that ${H'}$ is a graph of the smallest number of vertices in $\mathfrak{C}$.

    Thus, we conclude that for each root to leaf path of $T$, for any $s,t \in V(T)$ on this path, ${H'}_s$ and ${H'}_t$ belong to different equivalence classes of $\sim^r_{\mathcal{E}}$. 
    From Lemma~\ref{lem:equiv-one}, 
    the number of equivalence classes of $\sim^r_{\mathcal{E}}$ is at most $\texttt{index}_{\mathcal{E}}$.
    Since $(T,\beta)$ is a nice tree decomposition $T$ is a binary tree, therefore the number of nodes in $T$ is at most $2^{\texttt{index}_{\mathcal{E}} +1}$. The number of vertices in $H'$ is then the number of nodes in $T$ times $r+1$ (because each bag of $(T,\beta)$ has at most $r+1$ vertices).
\end{proof}

Henceforth, $\texttt{size}_{\mathcal{E}}:= 2^{\texttt{index}_{\mathcal{E}} +1} \cdot (r+1)$ and 
$\texttt{size}_{\mathcal{E}, \mathcal{G}}:= 2^{\texttt{index}_{\mathcal{E}, \mathcal{G}} +1} \cdot (r+1)$.

\noindent{\bf Finding representatives.}

\begin{lemma}\label{lem:find-rep}
    Let $\mathcal{G}$ be a graph class where every graph excludes a fixed graph $Q$ as a (topological) minor.
    Given an $r$-boundaried graph $H$ of treewidth $r$ on at least $\texttt{size}_{\mathcal{E}}+1$ (resp.~$\texttt{size}_{\mathcal{E},  \mathcal{G}} +1$) vertices,
    in $\Oh_{r,d, \mathcal{F}}(|H|)$ (resp~$\Oh_{r,d, \mathcal{F}, |Q|}(|V(H)|)$) time,
    one can find an $r$-boundaried graph $\tilde{H}$ such that $|V(\tilde{H})| \leq \texttt{size}_{\mathcal{E}}$ (resp.~$|V(\tilde{H})| \leq \texttt{size}_{\mathcal{E}, \mathcal{G}}$),
    and $H \sim^r_{\mathcal{E}} \tilde{H}$ (resp.~$H \sim^r_{\mathcal{E}} \tilde{H}$).
\end{lemma}
\begin{proof}
    We prove the lemma for $\sim^r_{\mathcal{E}}$. The proof for $\sim^r_{\mathcal{E}, \mathcal{G}}$ is analogous.
    Let $\mathfrak{R}$ be the collection of all $r$-boundaried graphs on at most $\texttt{size}_{\mathcal{E} }+1$ vertices. This can be computed in $\Oh_{\texttt{size}_{\mathcal{E}}}(1)$.
    For each graph $H' \in \mathfrak{R}$,
    we can compute $f_{H'}$ in time $\Oh_{\texttt{size}_{\mathcal{E}}}(1)$. 
    Thus, for any $H',H'' \in \mathfrak{R}$, 
    one can check whether $H' \sim^r_{\mathcal{E}} H''$ in time $\Oh_{\texttt{size}_{\mathcal{E}}}$.

    Let $(T,\beta)$ be a rooted nice tree decomposition of $H$ such that $\partial(H)$ is contained in the root of $(T,\beta)$.
    Let $t \in V(T)$ be the furthest from the root of $T$ such that $|V(H_t)| = \texttt{size}_{\mathcal{E}} +1$ (because $(T,\beta)$ is a nice tree decomposition, for any $s,t \in V(T)$ such that $t$ is a parent of $s$, $|V(H_s)| \leq |V(H_t)| \leq  |V(H_s)| +1$). 

    Let $H_t$ belong to the equivalence class $\mathfrak{C}$ of $\sim^r_{\mathcal{E}}$.
    Since $|V(H_t)| = \texttt{size}_{\mathcal{E}}+1$ and there exists a graph, say $\tilde{H_t}$ on at most $\texttt{size}_{\mathcal{E}}$ vertices in $\mathfrak{C}$, such that $H_t \sim^r_{\mathcal{E}} \tilde{H_t}$.
    Since $\mathfrak{R}$ contains all $r$-boundaried graphs on at most $\texttt{size}_{\mathcal{E}}$ vertices, it contains the graphs $H_t$ and $\tilde{H_t}$ too. Moreover, given $H_t$ can one find $\tilde{H_t}$ or some other equivalent and smaller graph from $\mathfrak{R}$, by comparing $H_t$ to each graph in $\mathfrak{R}$ on at most $\texttt{size}_{\mathcal{E}}$ vertices and checking if they are equivalent under $\sim^r_{\mathcal{E}}$ in $\Oh_{\texttt{size}_{\mathcal{E}}}(1)$ time.

    Having found $\tilde{H_t}$ with strictly smaller vertices than $H_t$, replace $H_t$ with $\tilde{H_t}$ in $H$. That is, set $H:= (H\setminus H_t) \oplus \tilde{H_t}$.
    The new graph $H$ is equivalent to the old $H$ because $\sim^r_{\mathcal{E}}$ is DP-friendly.
    Repeat the above procedure on the new $H$ until the size of $H$ is at most $\texttt{size}_{\mathcal{E}}$. Since each round takes $\Oh_{\texttt{size}_{\mathcal{E}}}(1)$ time and in each round the size of $H$ strictly decreases, the total time taken is $\Oh_{\texttt{size}_{\mathcal{E}}}(|V(H)|)$.
\end{proof}

\begin{proof}[Proof of Lemma~\ref{lem:prelim:replace}]  
    Using Lemma~\ref{lem:fdel-equiv},
    it is enough to find a graph $\widehat{H}$ of bounded size which is in the same equivalence class of $\sim^r_{\mathcal{E}}$ as $H$.
    From Lemmas~\ref{lem:equiv-one} (resp.~\ref{lem:equiv-two}),~\ref{lem:dp-friendly},~\ref{lem:rep-size} and~\ref{lem:find-rep}, such a graph $\widehat{H}$ can be found in linear time.
\end{proof}

\end{document}